\documentclass[acmsmall]{acmart}

\usepackage{quiver}
\usepackage{mathpartir}
\usepackage{enumitem}
\usepackage{wrapfig}
\usepackage{fancyvrb}
\usepackage{comment}
\usepackage{array}

\theoremstyle{remark}
\newtheorem{remark}{Remark}
\newenvironment{DIFnomarkup}{}{}



\setcopyright{rightsretained}
\acmDOI{10.1145/3704863}
\acmYear{2025}
\acmJournal{PACMPL}
\acmVolume{9}
\acmNumber{POPL}
\acmArticle{27}
\acmMonth{1}
\received{2024-07-11}
\received[accepted]{2024-11-07}

\newcommand{\To}{\Rightarrow}
\newcommand{\inl}{\mathsf{inl}}
\newcommand{\inr}{\mathsf{inr}}
\newcommand{\alt}{\mathrel{\bf \,\mid\,}}

\newcommand{\dyn}{{?}}
\newcommand{\nat}{\text{Nat}}

\newcommand{\ra}{\rightharpoonup}

\newcommand{\dyntodyn}{\dyn \ra\, \dyn}

\newcommand{\upc}[1]{\text{up}\,{#1}\,}
\newcommand{\dnc}[1]{\text{dn}\,{#1}\,}

\newcommand{\err}{\mho}
\newcommand{\zro}{\textsf{zro}}
\newcommand{\suc}{\textsf{suc}}
\newcommand{\lda}[2]{\lambda {#1} . {#2}}

\newcommand{\rel}{\mathrel{{\relbar}\hspace{-2px}{\relbar}\hspace{-2px}{\shortmid}\hspace{-2px}{\relbar}\hspace{-2px}{\relbar}}}

\newcommand{\ltdyn}{\sqsubseteq}
\newcommand{\gtdyn}{\sqsupseteq}

\newcommand{\later}{{\vartriangleright}}
\newcommand{\laterhs}{{\later}}
\newcommand{\type}{\texttt{Type}}
\newcommand{\lob}{\text{L\"{o}b}}
\newcommand{\tick}{\mathsf{tick}}
\newcommand{\nxt}{\mathsf{next}}
\newcommand{\fix}{\mathsf{fix}}

\newcommand{\calV}{\mathcal{V}}
\newcommand{\calE}{\mathcal{E}}

\newcommand{\tok}{\mathrel{\mathop{\to}\limits^{\textrm{\footnotesize k}}}}
\newcommand{\timesk}{\mathrel{\mathop{\times}\limits^{\textrm{k}}}}

\newcommand{\op}[1]{{#1}^{\textrm{op}}}

\newcommand{\Set}{\mathsf{Set}}
\newcommand{\ErrDom}{\mathsf{ErrDom}}
\newcommand{\PreDom}{\mathsf{PreDom}}

\newcommand{\qfix}{\texttt{qfix}}

\newcommand{\Clock}{\mathsf{Clock}}

\newcommand{\Nat}{\mathsf{Nat}}

\newcommand{\li}{L_\mho}

\newcommand{\ltls}{\ltdyn}
\newcommand{\bisim}{\approx}

\newcommand{\semltbad}{\lesssim}

\newcommand{\id}{\mathsf{id}}

\newcommand{\sem}[1]{\llbracket {#1} \rrbracket}

\newcommand{\push}{\text{push}}
\newcommand{\pull}{\text{pull}}

\newcommand{\pv}{P^{\mathcal{V}}}
\newcommand{\pe}{P^{\mathcal{E}}}
\newcommand{\ptbv}{\text{ptb}^\mathcal{V}}
\newcommand{\ptbe}{\text{ptb}^\mathcal{E}}

\newcommand{\arr}{\to}
\newcommand{\comp}{}

\newcommand{\ltsq}[2]{\mathrel{\ltdyn^{#1}_{#2}}}
\newcommand{\bisimsq}[2]{\mathrel{\bisim^{#1}_{#2}}}
\newcommand{\ltsqbisim}[2]{\mathrel{{\widetilde{\ltdyn}}^{#1}_{#2}}}
\newcommand{\ltbisim}{\mathrel{\widetilde{\ltdyn}}}
\newcommand{\vf}{\mathcal{V}_f}

\newcommand{\ef}{\mathcal{E}_f}

\newcommand{\morbisimid}[1]{\text{Endo}_\bisim({#1})}

\newcommand{\binrel}[1]{\mathbin{#1}}

\newcommand{\da}{\downarrow}

\newcommand{\piv}{\Pi^\mathcal{V}}

\newcommand{\upl}{\textsc{UpL}}
\newcommand{\upr}{\textsc{UpR}}
\newcommand{\dnl}{\textsc{DnL}}
\newcommand{\dnr}{\textsc{DnR}}

\newcommand{\delre}{\delta^{r,e}}
\newcommand{\delle}{\delta^{l,e}}
\newcommand{\delrp}{\delta^{r,p}}
\newcommand{\dellp}{\delta^{l,p}}

\newcommand{\qordeq}{\bisim}

\newcommand{\inat}{\text{Inj}_\mathbb{N}}
\newcommand{\itimes}{\text{Inj}_\times}
\newcommand{\iarr}{\text{Inj}_\to}

\newcommand{\delay}{\text{Delay}}
\newcommand{\ledelay}{\le^{\text{Del}}}
\newcommand{\bisimdelay}{\bisim^{\text{Del}}}
\newcommand{\tnow}{\mathsf{now}}
\newcommand{\tlater}{\mathsf{later}}
\newcommand{\Da}{\Downarrow}

\newcommand{\ptb}{i}

\begin{document}

\title{Denotational Semantics of Gradual Typing using Synthetic Guarded Domain Theory (Extended Version)}
\author{Eric Giovannini}
\affiliation{
  \department{Electrical Engineering and Computer Science}
  \institution{University of Michigan}
  \country{USA}
}
\email{ericgio@umich.edu}
\orcid{0009-0003-6871-1714}

\author{Tingting Ding}
\affiliation{
  \department{Electrical Engineering and Computer Science}
  \institution{University of Michigan}
  \country{USA}
}
\email{tingtind@umich.edu}
\orcid{0009-0000-5676-1886}

\author{Max S. New}
\affiliation{
  \department{Electrical Engineering and Computer Science}
  \institution{University of Michigan}
  \country{USA}
}
\email{maxsnew@umich.edu}
\orcid{0000-0001-8141-195X}

\begin{abstract}
  Gradually typed programming languages, which allow for soundly
  mixing static and dynamically typed programming styles, present a
  strong challenge for metatheorists. Even the simplest sound
  gradually typed languages feature at least recursion and errors,
  with realistic languages featuring furthermore runtime allocation of
  memory locations and dynamic type tags. Further, the desired
  metatheoretic properties of gradually typed languages have become
  increasingly sophisticated: validity of type-based equational
  reasoning as well as the relational property known as
  graduality. Many recent works have tackled verifying these
  properties, but the resulting mathematical developments are highly
  repetitive and tedious, with few reusable theorems persisting across
  different developments.

  In this work, we present a new denotational semantics for gradual
  typing developed using guarded domain theory. Guarded domain theory
  combines the generality of step-indexed logical relations for
  modeling advanced programming features with the modularity and
  reusability of denotational semantics. We demonstrate the
  feasibility of this approach with a model of a simple gradually
  typed lambda calculus and prove the validity of beta-eta equality
  and the graduality theorem for the denotational model. This model
  should provide the basis for a reusable mathematical theory of
  gradually typed program semantics. Finally, we have mechanized most
  of the core theorems of our development in Guarded Cubical Agda, a
  recent extension of Agda with support for the guarded recursive
  constructions we use.
\end{abstract}

\maketitle







\newif\ifdraft
\drafttrue
\renewcommand{\max}[1]{\ifdraft{\color{blue}[{\bf Max}: #1]}\fi}
\newcommand{\eric}[1]{\ifdraft{\color{orange}[{\bf Eric}: #1]}\fi}
\newcommand{\tingting}[1]{\ifdraft{\color{red}[{\bf Tingting}: #1]}\fi}

\section{Introduction}
  
\subsection{Gradual Typing and Graduality}
In programming language design, there is a tension between \emph{static} typing
and \emph{dynamic} typing disciplines. With static typing, the code is
type-checked at compile time, while in dynamic typing, the type checking is
deferred to run-time. Both approaches have benefits and excel in different
scenarios, with static typing offering compile-time assurance of a program's
type safety and type-based reasoning principles that justify program
optimizations, and dynamic typing allowing for rapid prototyping of a codebase
without committing to fixed type signatures.
Most languages choose between static or dynamic typing and as a result,
programmers that initially write their code in a dynamically typed language need
to rewrite some or all of their codebase in a static language if they would like
to receive the benefits of static typing once their codebase has matured.

\emph{Gradually typed languages} \cite{siek-taha06, tobin-hochstadt06} seek to
resolve this tension by allowing for both static and dynamic typing disciplines
to be used in the same codebase. These languages support smooth interoperability
between statically-typed and dynamically-typed styles, allowing the programmer to
begin with fully dynamically-typed code and \emph{gradually} migrate portions of the
codebase to a statically typed style without needing to rewrite the project in a
completely different language.

One of the fundamental theorems for gradually typed languages is
\emph{graduality}, also known as the \emph{dynamic gradual guarantee},
originally defined by Siek, Vitousek, Cimini, and Boyland
\cite{siek_et_al:LIPIcs:2015:5031, new-ahmed2018}.
Informally, graduality says that migrating code from dynamic to
static typing should only allow for the introduction of static or
dynamic type errors, and not otherwise change the behavior of the
program.
This is a way to capture programmer intuition that increasing type
precision corresponds to a generalized form of runtime assertions in
that there are no observable behavioral changes up to the point of the
first dynamic type error\footnote{once a dynamic type error is raised,
in languages where the type error can be caught, program behavior may
then further diverge, but this is typically not modeled in gradual
calculi.}.
Fundamentally, this property comes down to the behavior of
\emph{runtime type casts}, which implement these generalized runtime
assertions.

Additionally, gradually typed languages should offer some of the
benefits of static typing. While standard type soundness, that
well-typed programs are free from runtime errors, is not compatible
with runtime type errors, it is possible instead to prove that
gradually typed languages validate \emph{type-based reasoning}. For
instance, while dynamically typed $\lambda$ calculi only satisfy
$\beta$ equality for their type formers, statically typed $\lambda$
calculi additionally satisfy type-dependent $\eta$ properties that
ensure that functions are determined by their behavior under
application and that pattern matching on data types is safe and
exhaustive. A gradually typed calculus that validates these
type-dependent $\eta$ laws then provides some of the type-based
reasoning that dynamic languages lack.


\subsection{Denotational Semantics in Guarded Domain Theory}

Our goal in this work is to provide an \emph{expressive},
\emph{reusable}, \emph{compositional} semantic framework for defining
such well-behaved semantics of gradually typed programs.
Our approach to achieving this goal is to provide a compositional
\emph{denotational semantics}, mapping types to a kind of semantic
domain, terms to functions and relations such as term precision to
proofs of semantic relations between the denoted functions.
Since the denotational constructions are all syntax-independent, the
constructions we provide may be reused for similar languages. Since it
is compositional, components can be mixed and matched depending on
what source language features are present.

Providing a semantics for gradual typing is inherently complicated in
that it involves: (1) recursion and recursive types through the
presence of dynamic types, (2) effects in the form of divergence and
errors (3) relational models in capturing the graduality
property. Recursion and recursive types must be handled using some
flavor of domain theory. Effects can be modeled using monads in the
style of Moggi, or adjunctions in the style of
Levy \cite{MOGGI199155,39155,levy99,levystacks}. Relational properties and their verification
lead naturally to the use of reflexive graph categories or double
categories \cite{dunphythesis,new-licata18}.

The only prior denotational semantics for gradual typing was given by
New and Licata and is based on a classical Scott-style \emph{domain
theory} \cite{new-licata18}. The fundamental idea is to equip
$\omega$-CPOs with an additional ``error ordering'' $\ltdyn$ which
models the graduality ordering, and for casts to arise from
\emph{embedding-projection pairs}. Then the graduality property
follows as long as all language constructs can be interpreted using
constructions that are monotone with respect to the error ordering.
This framework has the benefit of being compositional, and was
expressive enough to be extended to model dependently typed gradual
typing \cite{gradualizing-cic}.
However, an approach based on classical domain theory has fundamental
limitations: domain theory appears incapable of modeling certain perversely
recursive features of programming languages such as dynamic type tag
generation and higher-order references \cite{Birkedal-Stovring-Thamsborg-2009}, which are commonplace in
real-world gradually typed systems as well as gradual calculi.
Our long-term goal is to develop a denotational approach that can
scale up to these advanced features, and so we must abandon classical
domain theory as the foundation in order to make progress. In this
work, we provide first steps towards this goal by adapting work on a
simple gradually typed lambda calculus to the setting of \emph{guarded domain
theory}, a task that already involves nontrivial issues.
Our goal is for this model to scale to dynamic type tag
generation and higher-order references in future work.

Guarded domain theory is the main denotational alternative to
classical domain theory that can successfully model the advanced
features we aim to support. While classical domain theory is based on modeling types as
ordered sets with certain joins, guarded domain theory is based on an
entirely different foundations, sometimes (ultra)metric spaces but
more commonly as ``step-indexed sets'', i.e., objects in the
\emph{topos of trees}
\cite{birkedal-mogelberg-schwinghammer-stovring2011}.  Such an object
consists of a family $\{X_n\}_{n \in \mathbb{N}}$ of sets along with
restriction functions $r_n : X_{n+1} \to X_n$ for all $n$.  (in
category theoretic terminology, these are presheaves on the poset of
natural numbers.)  We think of a $\mathbb{N}$-indexed set as an
infinite sequence of increasingly precise approximations to the true
type being modeled.
Key to guarded domain theory is that there is an operator
$\triangleright$ on step-indexed sets called ``later''. In terms of
sequences of approximations, the later operator delays the
approximation by one step. Then the crucial axiom of guarded domain
theory is that any guarded domain equation $X \cong F(\triangleright
X)$ has a unique solution. This allows guarded domain theory to model
essentially \emph{any} recursive concept, with the caveat that the
recursion is \emph{guarded} by a use of the later operator.


While guarded domain theory can be presented analytically using
ultrametric spaces or the topos of trees, in practice it is
considerably simpler to work \emph{synthetically} by working in a
non-standard foundational system such as \emph{guarded type
theory}. In guarded type theory later is taken as a primitive
operation on types, and we take as an axiom that guarded domain
equations have a (necessarily unique) solution. The benefit of this
synthetic approach is that when working in the non-standard
foundation, we don't need to model an object language type as a
step-indexed set, but instead simply as a set, and object-language
terms can be modeled in the Kleisli category of a simple monad defined
using guarded recursion. Not only does this make on-paper reasoning
about guarded domain theory easier, it also enables a simpler avenue
to verification in a proof assistant. Whereas formalizing analytic
guarded domain theory would require a significant theory of presheaves
and making sure that all constructions are functors on categories of
presheaves, formalizing synthetic guarded domain theory can be done by
directly adding the later modality and the guarded fixed point
property axiomatically.
%
%
In this paper, we work informally in a guarded type theory which is described in
more detail in Section \ref{sec:guarded-type-theory}.

%


\subsection{Adapting the New-Licata Model to the Guarded Setting}


Since guarded domain theory only provides solutions to guarded domain equations,
there is no systematic way to convert a classical domain-theoretic semantics to
a guarded one.  Classical domain theory has limitations in what it can model,
but it provides \emph{exact} solutions to domain equations when it applies. When
adapting the New-Licata approach to guarded domain theory, the presence of later
in the semantics makes it \emph{intensional}: unfolding the dynamic type
requires an observable computational step. For example, a function that pattern
matches on an element of the dynamic type and then returns it unchanged is
\emph{not} equal to the identity function, because it ``costs'' a step to
perform the pattern match.


The intensional nature of guarded domain theory leads to the main departure of
our semantics from the New-Licata approach. Because casts involve computational
steps, the graduality property must be insensitive to the steps taken by terms.
This means that the model must allow for reasoning \emph{up to weak
bisimilarity}, where two terms are weakly bisimilar if they differ only in their
number of computational steps. However, the weak bisimilarity relation is not
transitive in the guarded setting, which follows from a no-go theorem we
establish (Theorem \ref{thm:no-go}). The New-Licata proofs freely use
transitivity, and we argue in Section \ref{sec:towards-relational-model} that
some amount of transitive reasoning is necessary for defining a
syntax-independent model of gradual typing. As a result, the lack of
transitivity of weak bisimilarity presents a challenge in adapting the
New-Licata model to the guarded setting. Our solution is to work with a
combination of weak bisimilarity and step-sensitive reasoning, where the latter
notion \emph{is} transitive. To deal with the fact that casts take computational
steps, we introduce the novel concept of \emph{syntactic perturbations}, which
are explicit synchronizations that we manipulate syntactically. The combination
of step-sensitive reasoning with perturbations is transitive and is able to account
for the step-insensitive nature of the graduality property.


\subsection{Adequacy}


Having defined the model of gradual typing in guarded type theory, we need to
show that the guarded model induces a sensible set-theoretic semantics in
ordinary mathematics. We cannot hope to derive such a semantics for \emph{all}
types (e.g., the dynamic type), but for the subset of closed terms of base type,
we should be able to extract a well-behaved semantics for which the graduality
property holds.

More concretely, consider a gradually typed language whose only effects are
gradual type errors and divergence (errors arise from failing casts; divergence
arises because we can use the dynamic type to encode untyped lambda calculus;
see Section \ref{sec:GTLC}). If we fix a result type of natural numbers, a
``big-step'' semantics is a partial function from closed programs to either natural
numbers or errors:
\[ -\Downarrow : \{M \,|\, \cdot \vdash M : \nat \} \rightharpoonup \mathbb{N}
\cup \{\mho\} \] 
where $\mho$ is notation for a runtime type error. We write $M \Downarrow n$ and
$M\Downarrow \mho$ to mean this semantics is defined as a number or error, and
$M\Uparrow$ to mean the semantics is undefined, representing divergence.
A well-behaved semantics should then satisfy several properties. First, it
should be \emph{adequate}: natural number constants should evaluate to
themselves: $n \Downarrow n$. Second, it should validate type based reasoning.
To formalize type based reasoning, languages typically have an equational theory
$M \equiv N$ specifying when two terms should be considered equivalent.
Then we want to verify that the big step semantics respects this equational
theory: if closed programs $M \equiv N$ are equivalent in the equational theory
then they have the same semantics, $M \Downarrow n \iff N \Downarrow n$ and
$M\Uparrow \iff N \Uparrow$ and $M \Downarrow \mho \iff N \Downarrow \mho$.

We have to take care when defining such a big-step semantics in
guarded type theory, as a naive definition of termination is
incorrect: we can prove using guarded recursion that an infinite loop
terminates! This is because the definition of termination inside
guarded logic has a different meaning, essentially meaning that for
every finite number $n$, either the program terminates in $< n$ steps
or it takes $n+1$ steps, which is equivalent to a statement of
progress rather than termination. For this reason, we need a kind of
``escape hatch'' in our ambient type theory to leave the guarded
setting and make definitions that are interpreted as constructions in
ordinary mathematics. We incorporate this by using
\emph{clock-quantification} \cite{atkey-mcbride2013, kristensen-mogelberg-vezzosi2022}
in our type theory. In particular, the guarded type theory in which we work has
a notion of clocks that index all of our guarded constructions. This is
discussed in more detail in Section \ref{sec:big-step-term-semantics}.

To define the big-step semantics, we apply clock quantification to obtain
an interpretation of closed programs of base type as partial functions.
We describe the process in more detail in Section
\ref{sec:big-step-term-semantics}. The resulting term semantics is adequate, and
furthermore it validates the equational theory, since equivalent terms in the
equational theory denote equal terms in the big-step semantics.

Lastly, the semantics should be adequate for the graduality property. Graduality
is typically axiomatized by giving an \emph{inequational} theory called term
precision, where $M \ltdyn N$ roughly means that $M$ and $N$ have the same type
erasure and $M$ has at each point in the program a more precise/static type than
$N$. Then adequacy says that if $M$ and $N$ are whole programs returning type
$\nat$ and $M \ltdyn N$, then the denotations of $M$ and $N$ as big-step terms
should be related in the expected way: either $M$ errors, or $M$ and $N$ have
the same extensional behavior. That is, either $M\Downarrow \mho$ or $M
\Downarrow n $ and $N \Downarrow n$ or $M \Uparrow $ and $N
\Uparrow$\footnote{we use a slightly more complex definition of this relation in
our technical development below that is classically equivalent but
constructively weaker}.
To prove that the semantics is adequate for graduality, we again apply clock
quantification, this time to the relation that denotes the term precision
ordering. We show that a term precision ordering $M \ltdyn N$ implies the
corresponding ordering on the partial functions denoted by $M$ and $N$. The
details of the proof are given in Section \ref{sec:adequacy}.



\subsection{Contributions and Outline}

The main contribution of this work is a compositional denotational semantics in
guarded type theory for a simple gradually typed language that validates
$\beta\eta$ equality and satisfies a graduality theorem.
%
%
Within our semantics, the notion of syntactic perturbation is our most
significant technical contribution, allowing us to successfully adapt the
denotational approach of New and Licata to the guarded setting.
Most of our work has further been verified in Guarded Cubical Agda
\cite{veltri-vezzosi2020}, demonstrating that the semantics is readily
mechanizable. The Agda formalization is available online \cite{artifact}.
We provide an overview of the mechanization effort, and what
remains to be formalized, in Section \ref{sec:mechanization}.
%

%

The paper is laid out as follows:
\begin{enumerate}
\item In Section \ref{sec:GTLC} we fix our input language, a fairly
  typical gradually typed cast calculus.
\item In Section \ref{sec:concrete-term-model} we develop a
  denotational semantics in synthetic guarded domain theory for the
  \emph{terms} of the gradual lambda calculus.  The model is adequate
  and validates the equational theory, but it does not satisfy
  graduality. We use this to introduce some of our main technical
  tools: modeling recursive types in guarded type theory and modeling
  effects using call-by-push-value.
\item In Section \ref{sec:towards-relational-model} we show where the New-Licata
  classical domain theoretic approach fails to adapt cleanly to the guarded
  setting and explore the difficulties of proving graduality in an intensional
  model. We prove a no-go theorem about extensional, transitive relations,
  introduce the lock-step error ordering and weak bisimilarity relation, and
  motivate the need for perturbations.
\item In Section \ref{sec:concrete-relational-model} we describe the
  construction of the model in detail, and discuss the
  translation of the syntax and axioms of the gradually typed cast calculus into the model.
  Lastly, we prove that the model is adequate for the graduality property.
\item In Section \ref{sec:discussion} we discuss prior work on proving
  graduality, the partial mechanization of our results in Agda, and
  future directions for denotational semantics of gradual typing.
\end{enumerate}



\section{Syntactic Theory of Gradually Typed Lambda Calculus}\label{sec:GTLC}

Here we give an overview of a fairly standard cast calculus for
gradual typing along with its (in-)equational theory that capture our
desired notion of type-based reasoning and graduality. The main
departure from prior work is our explicit treatment of type precision
derivations and an equational theory of those derivations.

We give the basic syntax and select typing rules in
Figure~\ref{fig:gtlc-syntax}. We include a dynamic type, a type of
numbers, the call-by-value function type $A \ra A'$ and products.
We include a syntax for \emph{type precision} derivations $c : A
\ltdyn A'$; the typing is given in Figure~\ref{fig:gtlc-syntax}.
Any type precision derivation $c : A \ltdyn A'$ induces a pair of
casts, the upcast $\upc c : A \ra A'$ and the downcast $\dnc c : A' \ra
A$.
The syntactic intuition is that $c$ is a proof that $A$ is ``less dynamic'' than
$A'$. Semantically, $c$ denotes a relation between the denotations of $A$ and
$A'$ along with coercions back and forth; the upcast is (to a first-order)
a pure function while the downcast can fail.
These casts are inserted automatically in an elaboration from a
surface language. In this work, we are focused on semantic aspects and
so elide these standard details.
The syntax of precision derivations includes reflexivity $r(A)$ and
transitivity $cc'$ as well as monotonicity $c \ra c'$ and $c \times
c'$ that are \emph{covariant} in all arguments and finally generators
$\inat,\iarr,\itimes$ that correspond to the type tags of our dynamic
type.
We additionally impose an equational theory $c \equiv c'$ on the
derivations that implies that the corresponding casts are weakly
bisimilar in the semantics.
We impose category axioms for the reflexivity and
transitivity and functoriality for the monotonicity rules.
We note the following two admissible principles: any two derivations
$c,c' : A \ltdyn A'$ of the same fact are equivalent $c \equiv c'$ and
for any $A$, there is a derivation $\textrm{dyn}(A): A \ltdyn\dyn$. That is, $\dyn$ is the ``most dynamic'' type.

Our presentation of type precision is non-standard. In the usual
formulation, reflexivity and transitivity are admissible, and whenever
$A \ltdyn A'$, there is a unique precision derivation witnessing
this. 
%
These usual rules are given in Section \ref{sec:appendix-gtlc-syntax} in the
Appendix.
%
%
The reason for our presentation choice has to do with our
semantics, in which type precision derivations denote relations
between semantic elements of the types. In our semantics,
equivalent type precision derivations do not denote \emph{equal}
relations, but instead they denote relations that are only
\emph{quasi-equivalent}, i.e., if two terms are related by one then
they are related also by the other up to insertion of delays (see
Definition \ref{def:quasi-equivalent} for the details).
However, because all type precision derivations in our system are
equivalent, it is straightforward to define a translation from the
more standard system of type and term precision into ours, so
ultimately our graduality proofs still apply to the standard
formulation.

\begin{figure}
  \begin{mathpar}
    \begin{array}{rcl}
    \text{Types } A &::=& \nat \alt \,\dyn \alt A \ra A' \alt A \times A'\\
    \text{Type Precision } c &::=& r(A) \alt c c' \alt \iarr \alt \inat \alt \itimes \alt c \ra c' \alt c \times c'\\
    \text{Values } V &::=& x \alt \upc c V \alt \zro \alt \suc\, V \alt \lda{x}{M} \alt (V,V') \\ 
    \text{Terms } M,N &::=& \err\alt \upc c M \alt \dnc c M \alt \zro \alt \suc\, M \alt \lda{x}{M} \\ 
     &&\alt M\, N \alt (M,N) \alt \textrm{let } (x,y) = M \textrm{ in } N\\
    \text{Contexts } \Gamma &::= &\cdot \alt \Gamma, x : A \\
    \text{Ctx Precision } \Delta &::=& \cdot\alt \Delta,x:c
  \end{array}

  \inferrule
  {\Gamma \vdash M : A \and c : A \ltdyn A'}
  {\Gamma \vdash \upc c M : A'}

  \inferrule
  {\Gamma \vdash N : A' \and c : A \ltdyn A'}
  {\Gamma \vdash \dnc c N : A}

  \inferrule{}{\Gamma \vdash \mho : A}
  \end{mathpar}
  \begin{mathpar}
    \inferrule{}{r(A) : A \ltdyn A}\and
    \inferrule{c : A \ltdyn A' \and c' : A' \ltdyn A''}{cc' : A \ltdyn A''}\and
    \inferrule{}{\iarr \colon \dyn \ra \dyn \ltdyn \dyn}\and
    \inferrule{}{\inat \colon \nat \ltdyn \dyn}\and
    \inferrule{}{\itimes \colon \dyn \times \dyn \ltdyn \dyn}\and
    \inferrule{c_i : A_i \ltdyn A_i' \and c_o : A_o \ltdyn A_o'}{c_i \ra c_o : (A_i \ra A_o) \ltdyn (A_i' \ra A_o')}\and
    \inferrule{c_1 : A_1 \ltdyn A_1' \and c_2 : A_2 \ltdyn A_2'}{c_1 \times c_2 : (A_1 \times A_2) \ltdyn (A_1' \times A_2')}\and
     r(A)c \equiv c\and
     c \equiv cr(A')\and
     c(c'c'') \equiv (cc')c''\and
     r(A_i \ra A_o) \equiv r(A_i) \ra r(A_o)\and
     r(A_1\times A_2) \equiv r(A_1) \times r(A_2)\and
     (c_i \ra c_o)(c_i' \ra c_o')\equiv (c_ic_i' \ra c_oc_o') \and
     (c_1\times c_2)(c_1'\times c_2')\equiv (c_1c_1' \times c_2c_2')
  \end{mathpar}
  \caption{GTLC Cast Calculus Syntax, Type Precision Derivations and Precision Equivalence}
  \label{fig:gtlc-syntax}
\end{figure}

Next, we consider the axiomatic (in)equational reasoning principles
for terms: $\beta\eta$ equality and term precision in
Figure~\ref{fig:term-prec}.
We include standard CBV $\beta\eta$ rules for function and product
types. 
Next, we have \emph{term} precision, an extension of
type precision to terms.
The form of the term precision rule is $\Delta \vdash M \ltdyn M' : c$
where $\Delta$ is a context where variables are assigned to type
precision derivations.
The judgment is only well formed when every use of $x : c$ for $c : A
\ltdyn A'$ is used with type $A$ in $M$ and $A'$ in $M'$ and similarly
the output types match $c$.
We elide the congruence rules for every type constructor, e.g., if
$M \ltdyn M'$ and $N \ltdyn N'$ then $M\,N \ltdyn M'\,N'$.
With such congruence rules, reflexivity $M \ltdyn M$ is
admissible. Transitivity, on the other hand, is intentionally not
taken as a primitive rule, matching the original formulation of the
dynamic gradual guarantee \cite{siek_et_al:LIPIcs:2015:5031}.
We include a rule that says that equivalent type precision derivations
$c \equiv c'$ are equivalent for the purposes of term precision.
%
%
%
Finally, we include 4 rules for reasoning about casts. Intuitively
these say that the upcast is a kind of \emph{least upper bound} and
dually that the downcast is a \emph{greatest lower bound}.

As a higher-order gradually typed language, we inherently have to deal
with two effects: errors and divergence. Errors arise from failing
casts, e.g. casting a number to dynamic to a function
$\dnc{\iarr}\upc{\inat} x$. Divergence arises because our dynamic type
allows us to encode untyped lambda calculus, and so we can encode the
$\Omega$ term with the help of casts:
\[ \Omega = (\lambda (x:\dyn). (\dnc{\iarr} x)\, x)\, (\upc{\iarr}(\lambda (x:\dyn). (\dnc{\iarr} x)\, x)). \]

\begin{figure}
  \begin{mathpar}
  (\lambda x. M)(V) = M[V/x] \and (V : A \ra A') = \lambda x. V\,x\\

   \textrm{let } (x,y) = (V,V') \textrm{ in } N = N[V/x,V'/y] \and
   M[V:A\times A'/p] = \textrm{let } (x,y) = V \textrm{ in } M[(x,y)/p]


  \inferrule*[right=EquivTyPrec]
  {\Delta\vdash M \ltdyn M' : c \and c \equiv c'}
  {\Delta\vdash M \ltdyn M' : c'}

  \inferrule*[right=ErrBot]
  {}
  {\Delta \vdash \mho \ltdyn M : c}


  \inferrule*[right=UpL]
  {M \ltdyn M' : cc_r}
  {\upc {c} M \ltdyn M' : c_r}

  \inferrule*[right=UpR]
  {M \ltdyn M' : c_l}
  {M \ltdyn \upc {c} M' : c_lc}

  \inferrule*[right=DnL]
  {M \ltdyn M' : c_r}
  {\dnc {c} M \ltdyn M' : cc_r}

  \inferrule*[right=DnR]
  {M \ltdyn M' : c_lc}
  {M \ltdyn \dnc {c} M' : c_l}
  \end{mathpar}
  \caption{Equality and Term Precision Rules (Selected)}
  \label{fig:term-prec}
\end{figure}

Our goal in the remainder of this work is to develop compositional
denotational semantics of types, terms, type and term precision from
which we can easily extract a big step semantics that satisfies
graduality and respects the equational theory of the calculus.

\section{A Denotational Semantics for Types and Terms}\label{sec:concrete-term-model}

As a first step to giving a well-behaved semantics of GTLC, in this
section we define a simpler semantics that validates all of our
desired properties except for graduality: that is, it is adequate and
validates $\beta\eta$ equality. Our relational semantics will then be
a refinement of this simpler design. In the process we will introduce
some of the technical tools we will use in the refined semantics:
guarded type theory and call-by-push-value. While these tools can be
somewhat technical we emphasize that due to the use of synthetic
guarded domain theory, this initial semantics is essentially just a
``na\"\i ve'' denotational semantics of an effectful language arising
from a monad on sets. But since we are working in the non-standard
foundation of guarded type theory our notion of sets allows for the
construction of more sophisticated sets and monads than classical
logic allows.

%

\subsection{Guarded Type Theory}\label{sec:guarded-type-theory}

We begin with a brief overview of the background theory that we use in this work.
\emph{Synthetic guarded domain theory}, or SGDT for
short, is an axiomatic framework in which we can reason about non-well-founded
recursive constructions while abstracting away the specific details of
step-indexing that we would need to track if we were working analytically. This
allows us to avoid the tedious reasoning associated with traditional
step-indexing techniques. We provide a brief overview here; more details can be
found in \cite{birkedal-mogelberg-schwinghammer-stovring2011}.

SGDT offers a synthetic approach to domain theory that allows for guarded
recursion to be expressed syntactically via a type constructor $\later \colon
\type \to \type$ (pronounced ``later''). The use of a modality to express
guarded recursion was introduced by Nakano \cite{Nakano2000}.
Given a type $A$, the type $\later A$ represents an element of type $A$ that is
available one time step later. There is an operator $\nxt : A \to\, \later A$
that ``delays'' an element available now to make it available later. We will use
a tilde to denote a term of type $\later A$, e.g., $\tilde{M}$.


There is a \emph{guarded fixpoint} operator
\[ \fix : \forall T, (\later T \to T) \to T. \]
That is, to construct a term of type $T$, it suffices to assume that we have
access to such a term ``later'' and use that to help us build a term ``now''.
This operator satisfies the axiom that $\fix f = f (\nxt (\fix f))$.
In particular, $\fix$ applies when type $T$ is instantiated to a proposition 
$P : \texttt{Prop}$; in that case, it corresponds to $\lob$-induction.

\subsubsection{Ticked Cubical Type Theory}

Ticked Cubical Type Theory \cite{mogelberg-veltri2019} is an extension of
Cubical Type Theory \cite{CohenCoquandHuberMortberg2017} that has an additional
sort called \emph{ticks}. Ticks were originally introduced in
\cite{bahr-grathwohl-bugge-mogelberg2017}. A tick $t : \tick$ serves as evidence
that one unit of time has passed. In Ticked Cubical Type Theory, the type
$\later A$ is the type of dependent functions from ticks to $A$. The type $A$ is
allowed to depend on $t$, in which case we write $\later_t A$ to emphasize the
dependence.



The rules for tick abstraction and application are similar to those of ordinary
$\Pi$ types. A context now consists of ordinary variables $x : A$ as well as
tick variables $t : \tick$. The presence of the tick variable $t$ in context
$\Gamma, (t : \tick), \Gamma'$ intuitively means that the values of the
variables in $\Gamma$ arrive ``first'', then one time step occurs, and then the
values of the variables in $\Gamma'$ arrive.
The abstraction rule for ticks states that if in context $\Gamma, t : \tick$
the term $M$ has type $A$, then in context $\Gamma$ the term $\lambda t.M$ has
type $\later_t A$.
Conversely, if we have a term $M$ of type $\later A$, and we have available in
the context a tick $t' : \tick$, then we can apply the tick to $M$ to get a
term $M[t'] : A[t'/t]$. However, there is an important restriction on when we
are allowed to apply ticks: To apply $M$ to tick $t$, we require $M$ to be
well-typed in the prefix of the context occurring before the tick $t$. That is,
all variables mentioned in $M$ must be available before $t$. This ensures that
we cannot, for example, define a term of type $\later \laterhs A \to\, \laterhs
A$ via repeated tick application.
%
For the sake of brevity, we will also write tick application as $M_t$.
The constructions we carry out in this paper will take place in a guarded type
theory with ticks.
\subsection{A Call-by-Push-Value Model}


Our denotational model of GTLC must account for the fact that the
language is call-by-value with non-trivial effects (divergence and
errors).  To do this, we formulate the model using the categorical
structures of Levy's \emph{call-by-push-value}
\cite{levy99}. Call-by-push-value is a categorical model of effects
that is a refinement of Moggi's monadic semantics
\cite{MOGGI199155}. While Moggi works with a (strong) monad $T$ on a
category $\mathcal V$ of pure morphisms, call-by-push-value works
instead with a decomposition of such a monad into an adjunction
between a category $\mathcal V$ of ``value types'' and pure morphisms
and a category $\mathcal E$ of ``computation types'' and homomorphisms.
The monad $T$ on $\calV$ is decomposed into a left adjoint
$F : \calV \to \calE$ constructing the type $F A$ of computations that
return $A$ values and right adjoint $U : \calE \to \calV$
constructing the type $U B$ of values that are first-class suspended
computations of type $B$. The monad $T$ can then be reconstituted as
$T = UF$, but the main advantage is that many constructions on
effectful programs naturally decompose into constructions on
value/computation types. Most strikingly, the CBV function type $A
\rightharpoonup A'$ decomposes into the composition of three
constructions $U(A \to F A')$ where $\arr : \calV^{op} \times \calE
\to \calE$ is a kind of mixed kinded function type. This decomposition
considerably simplifies treatment of the CBV function
type. Additionally, prior work on gradual typing semantics has shown
that the casts of gradual typing are naturally formulated in terms of
pairs of a pure morphism and a homomorphism, as we explain below
\cite{new-licata-ahmed2019}.

For our simple denotational semantics, the category $\calV$ is simply
the category $\Set$ of sets and functions. Then because GTLC is
call-by-value, all types will be interpreted as objects of this
category, i.e., sets. More interesting are the computation
types. These will be interpreted in a category of sets equipped with
algebraic\footnote{The $\theta$ structure is not algebraic in the
strictest sense since it does not have finite arity.} structure with
morphisms being functions that are homomorphic in this structure.
This algebraic structure needs to model the effects present in our
language, in this case these are errors as well as taking
computational steps (potentially diverging by taking computational
steps forever).
We call these algebraic structures simple error domains to distinguish
them from the error domains we define later to model graduality:
\begin{definition}
  A (simple) error domain $B$ consists of
  \begin{enumerate}
  \item A carrier set $UB$
  \item An element $\mho_{B} : UB$ representing error
  \item A function $\theta_B : \laterhs UB \to UB$ modeling a computational step
  \end{enumerate}
  A homomorphism of error domains $\phi : B \multimap B'$ is a
  function $\phi : UB \to UB'$ that preserves $\mho$ and $\theta$, i.e.,
  $\phi(\mho_B) = \mho_{B'}$ and $\phi(\theta_{B}(\tilde{x})) =
  \theta_{B'}(\later (\phi(\tilde{x})))$.  Simple error domains and
  homomorphisms assemble into a category $\ErrDom$ with a forgetful
  functor $U : \ErrDom \to \Set$ taking the underlying set and
  function.
\end{definition}
Here is our first point where we utilize guarded type theory: rather
than simply being a function $UB \to UB$, the ``think'' map $\theta$
takes an element \emph{later}. This makes a major difference, because
the structure of a think map combined with the guarded fixpoint
operator allows us to define recursive elements of $UB$ in that any
function $f : UB \to UB$ has a ``quasi-fixed point'' $\textrm{qfix}(f)
= \fix(f \circ \theta_B)$ satisfying the quasi-fixed point property:
\[ \qfix(f) = f(\delta_B(\qfix f)) \]
where $\delta_B = \theta_B \circ \nxt$ is a map we call the ``delay''
map which trivially delays an element now to be available later.  As
an example, we can define $\Omega_B = \qfix(\id) : UB$, the
``diverging element'' that ``thinks forever'' in that $\Omega_B =
\delta_B(\Omega_B)$. We call this a ``quasi'' fixed point because it is
\emph{nearly} a fixed point except for the presence of the delay map
$\delta_B$, which is irrelevant from an extensional point of view
where we would prefer to ignore differences in the number of steps
that computations take.


Then to model effectful programs, we need to construct a left adjoint
$\li : \Set \to \ErrDom$ to $U$ which constructs the ``free error
domain'' on a set, which models an effectful CBV computation.
We base the construction on the \emph{guarded
lift monad} described in \cite{mogelberg-paviotti2016}. Here, we augment the
guarded lift monad to accommodate the additional effect of failure.
\begin{definition}
  For a set $A$, we define the (carrier of) \emph{free error domain} $U(\li A)$ as the unique solution to the guarded domain equation:
  \[ U(\li A) \cong A + 1\, + \laterhs U(\li A). \]
  For which we use the following notation for the three constructors:
  \begin{enumerate}
  \item $\eta \colon A \to U(\li A)$
  \item $\mho \colon U(\li A)$
  \item $\theta \colon \laterhs U(\li A) \to U(\li A)$
  \end{enumerate}
  Here $\mho, \theta$ provide the error domain structure for $\li A$,
  and this has the universal property of being free in that any
  function $f : U(\li A) \to UB$ uniquely extends to a homomorphism
  $f^\dagger : \li A \multimap B$ satisfying $Uf^\dagger \circ \eta =
  f$. In other words, the free error monad is left adjoint $\li \dashv
  U$ to the forgetful functor, with $\eta$ the unit of the adjunction.
  We will write $\theta_t(\dots)$ to mean $\theta (\lambda t. \dots)$.
\end{definition}
%
An element of $\li A$ should be viewed as a computation that can either (1)
return a value (via $\eta$), (2) raise an error and stop (via $\mho$), or (3)
take a computational step (via $\theta$).
The operation $f^\dagger$ can be defined by guarded recursion in the
apparent way, and proven to satisfy the freeness property by L\"ob
induction.
%
%




The free error domain is essential to modeling CBV computation: a term
$\Gamma \vdash M : A$ in GTLC will be modeled as a function $M :
\sem{\Gamma} \to U\li \sem{A}$ from the set denoted by $\Gamma$ to the
underlying set of the free error domain on the set denoted by $A$,
just as in a monadic semantics with monad $U\li$.

\subsection{Modeling the Type Structure}\label{sec:dynamic-type}
Much of the type structure of GTLC is simple to model: $\nat$ denotes
the set of natural numbers and $\times$ denotes the Cartesian product
of sets. To model the CBV function type $A \rightharpoonup A'$ we use
the CBPV decomposition $U(\sem{A} \to \li \sem{A'})$ where $A \to B$
is the obvious point-wise algebraic structure on the set of functions
from $A$ to $UB$.

The final type to model is the the dynamic type $\dyn$.  In the style
of Scott, a classical domain theoretic model for $\dyn$ would be given
by solving a domain equation
\[ D \cong \Nat + (D \times D) + U(D \to \li D), \]
representing the fact that dynamically typed values can be numbers,
pairs or CBV functions. This equation does not have inductive or
coinductive solutions due to the negative occurrence of $D$ in the
domain of the function type. In classical domain theory this is
avoided by instead taking the initial algebra in a category of
embedding-projection pairs, by which we obtain an exact solution to
this equation. However, in guarded domain theory we instead replace this domain equation with a guarded domain equation.
Consider the following similar-looking equation:
\begin{equation}\label{eq:dyn}
D \cong \Nat + (D \times D)\, + \laterhs U(D \to \li D).
\end{equation}
Since the negative occurrence of $D$ is guarded under a later, we can
solve this equation by a mixture of inductive types and guarded fixed
points. Specifically, consider the parameterized inductive type
\[ D'\, X := \mu T. \Nat + (T \times T) + X. \]
Then we can construct $D$ as the unique solution to the guarded domain equation
\[ D \cong D'(\laterhs U(D \to \li D)) \]
which is a guarded equation in that all occurrences of $D$ on the right hand side are guarded by the $\laterhs$.
Expanding out the guarded fixed point and least fixed point property
gives us that $D$ satisfies Equation $\ref{eq:dyn}$. 
%
We then refer to the three injections for numbers, pairs and functions
as $\inat, \itimes, \iarr$ respectively.

\subsection{Term Semantics}\label{sec:term-interpretation}

We now extend our semantics of types to a semantics of terms of
GTLC. As mentioned above, terms $x:A_1,\ldots \vdash M : A$ are
modeled as effectful functions $\sem{M}: \sem{A_1}\times \cdots \to U\li\sem{A}$
and values $x:A_1,\ldots \vdash V : A$ are also modeled as pure
functions $\sem{V} : \sem{A_1}\times\cdots \to \sem{A}$.
The interpretation of the typical CBV features of GTLC into a model of CBPV is
standard \cite{levy99}; the only additional part we need to account for
is the semantics of casts, shown in Figure \ref{fig:term-semantics}.
%
%
Following the CBPV treatment of gradual typing given by
New-Licata-Ahmed, an upcast $\upc c$ where $c : A \ltdyn A'$ is
interpreted as a pure function $\sem{\upc c} :\sem{A} \to \sem{A'}$ whereas a
downcast $\dnc c$ is interpreted as a homomorphism $\sem{\dnc c} :
\li\sem{A'} \multimap \li\sem{A}$.
We define the semantics of upcasts and downcasts simultaneously by
recursion over type precision derivations.
The reflexivity casts are given by identities, and the transitivity
casts by composition. The upcasts for products, and the downcasts for
functions are given by the functorial actions of the type constructors
on the casts for the subformulae.
The \emph{up}casts for functions and similarly the \emph{down}casts
for products, on the other hand, are more complex.
For these, we use not the ordinary functorial action of type
constructors, but an action we call the \emph{Kleisli} action.

The Kleisli action of type constructors can be defined in any CBPV
model. First, we can define the Kleisli category of value types
$\calV_k$ to have value types as objects, but have morphisms
$\calV_k(A,A') = \calE(FA,FA')$. Similarly, we can define a Kleisli
category of computation types $\calE_k$ to have computation types as
objects, but have morphisms $\calE_k(B,B') = \calV(UB,UB')$. The
intuition is that these are the ``effectful'' morphisms between
value/computation types, whereas the morphisms of $\calV$ and $\calE$ are
pure morphisms and homomorphisms, respectively. Then every CBPV type constructor $F,U,\times,\to$
has in addition to its usual functorial action on pure morphisms/homomorphisms
an action in each argument on Kleisli morphisms. So for
instance, the function construction is functorial in pure morphisms/homomorphisms
$\to : \op{\calV} \times \calE \to \calE$ but on Kleisli
morphisms what we have is for each computation type $B$ a functorial
action $-\to B : \calV_k^{op} \to \calE_k$ given by $-\tok B$ on
objects and for each value type a functorial action $A \tok - :
\calE_k \to \calE_k$. Similarly we get two actions $A \timesk -$ and
$-\timesk A'$ on Kleisli morphisms between value types. In these two
cases we do not get a bifunctorial action because the two Kleisli
actions do not commute past each other, i.e., in general $(A_l'
\timesk f_r \circ f_l \timesk A_r) \neq (f_l \timesk A_r' \circ A_l
\timesk f_r)$.
%
The full definitions of the Kleisli actions are included in Appendix
\ref{sec:kleisli-actions}.
%
%

Lastly, we have the upcasts and downcasts for the injections into the dynamic type, which
are the core primitive casts. The upcasts are simply the injections
themselves, except the function case which must include a $\nxt$ to
account for the fact that the functions are under a later in the
dynamic type. The downcasts are similar in that on values
they pattern match on the input and return it if it is of the correct
type, otherwise erroring. Again, the function case is slightly
different in that if the input is in the function case, then it is
actually only a function available later, and so we must insert a
``think'' in order to return it.

\begin{figure*}
  \begin{minipage}{0.3\textwidth}
    \begin{footnotesize}
  \begin{align*}
    \sem{\upc{r(A)}} &= \id_{\sem{A}} \\
    \sem{\upc{(c \comp c')}} &= \sem{\upc{c'}} \circ \sem{\upc{c}} \\
    \sem{\upc{(c_i \ra c_o)}} &= (\sem{\dnc {c_i}} \tok \li\sem{A_o'})\circ (\sem{A_i} \tok U\li(\sem{\upc {c_o}}))\\
    \sem{\upc{(c_1 \times c_2)}} &= \sem{\upc{c_1}} \times \sem{\upc{c_2}}\\
    \sem{\upc{\inat}} &= \inat\\
    \sem{\upc{\itimes}} &= \itimes\\
    \sem{\upc{\iarr}} &= \iarr \circ \nxt\\
    %
    \sem{\dnc{r(A)}} &= \id_{\li \sem{A}} \\
    \sem{\dnc{(c \comp c')}} &= \sem{\dnc{c}} \circ \sem{\dnc{c'}} \\
    \sem{\dnc{(c_i \ra c_o)}} &= U(\sem{\upc {c_i}} \to \sem{\dnc {c_o}})\\
    \sem{\dnc{(c_1 \times c_2)}} &= (\sem{\dnc{c_1}} \timesk \sem{A_2}) \circ (\sem{A_1'} \timesk \sem{\dnc{c_2}}) \\
    \sem{\dnc{\inat}} &= (
      \lambda V_d . \text{case $({V_d})$ of }
      \{ \inat\,n \to \eta\, n
         \alt \text{otherwise} \to \mho \})^\dagger \\
    \sem{\dnc{\itimes}} &= (
      \lambda V_d . \text{case $({V_d})$ of }
      \{ \itimes(d_1, d_2) \to \eta\, (d_1, d_2)
         \alt \text{otherwise} \to \mho \})^\dagger \\
    \sem{\dnc{\iarr}} &= (
      \lambda V_d . \text{case $({V_d})$ of }
      \{ \iarr \tilde{f} \to \theta_t\, (\eta (\tilde{f}_t))
         \alt \text{otherwise} \to \mho \})^\dagger \\
  \end{align*}
    \end{footnotesize}
  \end{minipage}
  \caption{Semantics of casts}
  \label{fig:term-semantics}
\end{figure*}

\subsection{Extracting a Well-Behaved Big-Step Semantics}\label{sec:big-step-term-semantics}


We now define from the above term model a ``big-step'' semantics. More
concretely, our goal is to define a partial function 
$-\Downarrow : \{M \,|\, \cdot \vdash M : \nat \} \rightharpoonup \mathbb{N} + {\mho}$.

To begin, we must first discuss another aspect of guarded type theory. The
constructions of guarded type theory ($\later$, $\nxt$, and $\fix$) we have been
using may be indexed by objects called \emph{clocks} \cite{atkey-mcbride2013}. A clock intuitively serves
as a reference relative to which steps are counted. For instance, given a clock
$k$ and type $X$, the type $\later^k X$ represents a value of type $X$ one unit
of time in the future according to clock $k$.
Up to this point, all uses of guarded constructs have been indexed by a single
clock $k$. In particular, we have for each clock $k$ a type $\li^k X$. The
notion of \emph{clock quantification} allows us to pass from the guarded setting
to the usual set-theoretic world. Clock quantification enables us to encode
coinductive types using guarded recursion. That is, given a type $A$ with a free
variable $k : \Clock$, we can form the type $A^{gl} := \forall k. A$ (here, $gl$ stands for ``global'').

A functor $F \colon \mathsf{Type} \to \mathsf{Type}$ is said to \emph{commute
with clock quantification} if for all $X : \Clock \to \mathsf{Type}$, the
canonical map $F(\forall k. X_k) \to \forall k. F(X_k)$ is an equivalence \cite{kristensen-mogelberg-vezzosi2022},
where we write $X_k$ for the application of clock $k$ to $X$.
Given a functor $F$ that commutes with clock quantification, if $A$ is a guarded
recursive type satisfying $A \cong F(\later^k A)$, then the type $A^{gl}$ has a
final $F$-coalgebra structure (see Theorem 4.3 of Kristensen et al. \cite{kristensen-mogelberg-vezzosi2022}). 

Consider the functor $F(X) = \mathbb{N} + {\mho} + X$. By the definition of the
guarded lift, we have that
$\li^k \mathbb{N} 
  \cong \mathbb{N} + {\mho} + \later^k (\li \mathbb{N}) 
  \cong F(\later \li^k \mathbb{N})$.
Furthermore, it can be shown that the functor $F$ commutes with clock
quantification. Thus, the type $\li^{gl} \mathbb{N} := \forall k. \li^k
\mathbb{N}$ is a final $F$-coalgebra. In fact, $\li^{gl}$ is isomorphic to a
more familiar coinductive type.
Given a type $Y$, Capretta's \emph{delay monad} \cite{lmcs:2265},
which we write as $\text{Delay}(Y)$, is defined as the coinductive type
generated by $\tnow : Y \to \delay(Y)$ and $\tlater : \delay(Y) \to \delay(Y)$.
We have that $\delay (\mathbb{N} + {\mho})$ is a final coalgebra for the functor
$F$ defined above. This means $\li^{gl} \mathbb{N}$ is isomorphic to
$\delay(\mathbb{N} + {\mho})$. We write $\alpha$ to denote this isomorphism.

Given $d : \delay(\mathbb{N} + {\mho})$, we define a notion of termination in
$i$ steps with result $n_? \in \mathbb{N} + {\mho}$, written $d \da^i n_?$.
We define $d \da^i n_?$ inductively by the rules
$(\tnow\, n_?) \da^0 n_?$ and $(\tlater\, d) \da^{i+1} n_?$ if $d \da^i n_?$.
From this definition, it is clear that if there is an $i$
such that $d \da^i n_?$, then this $i$ is unique. 
Using this definition of termination, we define a partial function 
$-\Downarrow^\text{Delay} \colon \delay (\mathbb{N} + {\mho}) \rightharpoonup \mathbb{N} + {\mho}$
by existentially quantifying over $i$ for which $d \da^i n_?$. That is, we
define $d \Downarrow^\text{Delay} \,= n_?$ if $d \da^i n_?$ for some (necessarily
unique) $i$, and $d \Downarrow^\text{Delay}$ is undefined otherwise.

Now consider the global version of the term semantics, $\sem{\cdot}^{gl} \colon
\{M \,|\, \cdot \vdash M : \nat \} \to \li^{gl} \mathbb{N}$, defined by
$\sem{M}^{gl} := \lambda (k : \Clock) .\sem{M}$. Composing this with the
isomorphism $\alpha : \li^{gl} \mathbb{N} \cong \delay(\mathbb{N} + {\mho})$ and
the partial function $-\Downarrow^\text{Delay}$ yields the desired partial
function
$-\Downarrow \colon \{M \,|\, \cdot \vdash M : \nat \} \rightharpoonup \mathbb{N} + {\mho}$.
It is then straightforward to see that this big-step semantics respects the
equational theory and that the denotations of the diverging term $\Omega$, the
error term $\mho$, and a syntactic natural number value $n$ are all distinct.

\section{Challenges to a Guarded Model of Graduality}\label{sec:towards-relational-model}


Now that we have seen how to model the terms without regards to
graduality, we next turn to how to enhance this model to provide a
compositional semantics that additionally satisfies graduality. By a
compositional semantics, we mean that we want to provide a
compositional interpretation of \emph{type} and \emph{term precision}
such that we can extract the graduality relation from our model of
term precision.

\subsection{A First Attempt: Modeling Term Precision with Posets}

In prior work, New and Licata model the term precision relation by
equipping $\omega$-CPOs with a second ordering $\ltdyn$ they call the
``error ordering'' such that a term precision relationship at a fixed
type $M \ltdyn M' : r(A)$ implies the error ordering $\sem{M} \ltdyn
\sem{M'}$ in the semantics. Then the error ordering on the free error
domain $\li A$ is based on the graduality property: essentially either
$\sem{M}$ errors or both terms diverge or both terms terminate with
values related by the ordering relation on $A$.
The obvious first try at a guarded semantics that models graduality
would be to similarly enhance our value types and computation types to
additionally carry a poset structure. For most type formers this
ordering can be defined as lifting poset structures on the
subformulae; the key design choice is in the ordering on the free
error domain $\li A$.

The ordering on $\li A$ should model the graduality property, but how
to interpret this is not so straightforward. The graduality property
as stated mentions divergence, but we only work indirectly with
diverging computations in the form of the think maps. The closest
direct encoding of the graduality property is an ordering we call the
\emph{step-insensitive error ordering}, defined by guarded recursion
in Figure~\ref{fig:step-insensitive-error-ordering}.
Firstly, to model graduality, $\mho$ is the least element. Next, two
returned values are related just when they are in the ordering $\ltdyn$
on $A$.
Two thinking elements are related just when they are related later.
The interesting cases are then those where one side is thinking and
the other has completed to either an error or a value.
Since the graduality property is \emph{extensional}, i.e., oblivious
to the number of steps taken, it is sensible to say that if one side
has terminated and the other side is thinking, that we must require
the thinking side to eventually terminate with a related behavior,
which is the content of the final three cases of the definition.
\begin{figure*}
      \begin{align*}
        \mho \semltbad l &\text{ iff } \top \\
        \eta\, x \semltbad \eta\, y &\text{ iff } 
            x \ltdyn y \\
        \theta\, \tilde{l} \semltbad \theta\, \tilde{l'} &\text{ iff } 
            \later_t (\tilde{l}_t \semltbad \tilde{l'}_t) \\
        \theta\, \tilde{l} \semltbad \mho &\text{ iff } \exists n. \theta\, \tilde{l} = \delta^n(\mho) \\
        \theta\, \tilde{l} \semltbad \eta\, y &\text{ iff } \exists n. \exists x \ltdyn y.
            (\theta\, \tilde{l} = \delta^n(\eta\, x)) \\
        \eta\, x \semltbad \theta\, \tilde{l'} &\text { iff }
            \exists n. \exists y \gtdyn x. (\theta\, \tilde{l'} = \delta^n (\eta\, y))
    \end{align*}
    \caption{Step-insensitive error ordering}
    \label{fig:step-insensitive-error-ordering}
\end{figure*}

This definition of $\semltbad$ is sufficient to model the graduality
property in that if $l \semltbad l'$ for $l,l' : \li \mathbb{N}$ where
$\mathbb{N}$ is the flat poset whose ordering is equality, then the
big-step semantics we derive for $l, l'$ does in fact satisfy the
intended graduality relation.
However, there is a major issue with this definition: it's not a poset
at all: it is not anti-symmetric, but more importantly it is
\emph{not} transitive!
To see why, we observe the following undesirable property of any
relation that is transitive and ``step-insensitive'' on one side:
\begin{theorem}\label{thm:no-go}
  Let $R$ be a binary relation on the free error domain $U(\li A)$. If
  $R$ satisfies the following properties
  \begin{enumerate}
  \item Transitivity
  \item $\theta$-congruence: If $\later_t (\tilde{x}_t \binrel{R} \tilde{y}_t)$, then $\theta(\tilde{x}) \binrel{R} \theta(\tilde{y})$.
  \item Right step-insensitivity: If $x \binrel{R} y$ then $x \binrel{R} \delta y$.
  \end{enumerate}
  Then for any $l : U(\li A)$, we have $l \binrel{R} \Omega$. If $R$ is left
  step-insensitive instead then $\Omega \binrel{R} x$.
\end{theorem}
\begin{proof}
  By L\"ob induction: we assume that $l$ is related to $\Omega$ later, which implies that
  $\theta (\nxt\, l) \binrel{R} \theta (\nxt\, \Omega)$. Then observe that 
  $l \binrel{R} \theta (\nxt\, l) \binrel{R} \theta (\nxt\, \Omega) = \Omega$.
  See the appendix for more details.
\end{proof}
\begin{corollary}
  Let $R$ be a binary relation on $U(\li A)$. If $R$ satisfies
  transitivity, $\theta$-congruence and left and right
  step-insensitivity, then $R$ is the total relation: $\forall x, y. x
  \binrel{R} y$.
\end{corollary}
\begin{proof}
  $x \binrel{R} \Omega$ and $\Omega \binrel{R} y$ by the previous
  theorem. Then by transitivity $x \binrel R y$.
\end{proof}

This in turn implies $\semltbad$ is not transitive, since it is a
$\theta$-congruence and step-insensitive on both sides, but not
trivial: e.g., for the flat poset on $\mathbb N$, $\eta 0 \semltbad
\eta 1$ is false.
This shows definitively that we cannot provide a compositional
semantic graduality proof based on types as posets if we use this
step-insensitive ordering on $\li A$.

\subsection{Why Transitivity is Essential}

At this point we might try to weaken our initial attempt at
providing a poset semantics to a semantics where types are equipped
with merely a reflexive relation and continue using the
step-insensitive ordering. However, some level of transitivity is
absolutely essential to providing compositional reasoning, and
transitivity is used pervasively in the prior work by New and Licata.
The reason that it is essential becomes clear by examining the details of the
graduality proof, so first we expand on how we prove graduality compositionally.
First, we need to extend our compositional semantics of types and terms to a
compositional semantics of \emph{type precision derivations} and \emph{term
precision}.
Sticking to a poset-based semantics, if our value types are interpreted as
posets, then we want our precision derivations $c : A \ltdyn A'$ to be
interpreted as a \emph{relation} between $A$ and $A'$. It is also natural to
ask that the relation $c$ interact with the orderings on $A$ and $A'$. More
specifically, if $x' \ltdyn_{A} x$ and $x \mathrel{c} y$, then $x' \mathrel{c}
y$. Similarly, if $x \mathrel{c} y$ and $y \ltdyn_{A'} y'$, then $x \mathrel{c}
y'$. We summarize this requirement by saying $c$ is \emph{downward-closed} in
$A$ and $\emph{upward-closed}$ in $A'$. We will see below why this requirement
is necessary for the graduality proof.


We call such a relation between $A$ and $A'$ a \emph{poset relation},
and we write $c : A \rel A'$ to mean $c$ is a poset relation downward
closed in $A$ and upward closed in A'. The paradigmatic example of a
poset relation is given by the ordering $\ltdyn_A$. We call this the
reflexive poset relation and denote it by $r(A) : A \rel A$. This is a
poset relation because the ordering $\ltdyn_A$ is transitive.
%
Given two poset relations $c : A_1 \rel A_2$ and $c' : A_2 \rel A_3$, we can
form their \emph{composition} in the usual manner, i.e., $x \binrel{c \comp c'} y$ if and
only if there is an element $z : A_2$ with $x \binrel{c} z$ and $z \binrel{c'} y$.



Poset relations are then used in the semantics of \emph{term}
precision. For closed programs $M \ltdyn M' : c$ where $c : A \ltdyn
A'$, the terms denote elements of the free error domain
$\sem{M},\sem{M'} : U\li \sem{A}$, and we model the term precision
semantically as the relation holding between those denoted elements
$\sem{M} \binrel{(U\li\sem{c})} \sem{M'}$. Here $U\li\sem{c}$ is a
relational lifting of $U\circ\li$, which extends $\sem{c}$ to account
for error and computational steps.
More generally, $M$ and $M'$ can be \emph{open} terms, in which case
they denote functions, not just elements. We model that as not simply
elements being related but in a typical logical-relations style as
\emph{preserving} relatedness: if they are passed in inputs related by
the domain relations then their output is related by the codomain
relation.
To capture this preservation of relations property, we use the notion of a
\emph{square}, which specifies a relation between two monotone functions. The
definition of square is as follows. Let $A_i, A_o, A_i', A_o'$ be
partially-ordered sets. Let $c_i : A_i \rel A_i'$ and $c_o : A_o \rel A_o'$ be
poset relations, and let $f : A_i \to A_o$ and $g : A_i' \to A_o'$ be monotone
functions. We say that $f \ltsq{c_i}{c_o} g$ if for all $x : A_i$ and $y : A_i'$
with $x \binrel{c_i} y$, we have $f(x) \binrel{c_o} g(y)$. We call this a
\emph{square} because we visualize this situation as follows:
%
\[\begin{tikzcd}[ampersand replacement=\&]
  {A_i} \& {A_i'} \\
  {A_o} \& {A_o'}
  \arrow["{c_i}", "\shortmid"{marking}, no head, from=1-1, to=1-2]
  \arrow["f"', from=1-1, to=2-1]
  \arrow["g", from=1-2, to=2-2]
  \arrow["{c_o}"', "\shortmid"{marking}, no head, from=2-1, to=2-2]
\end{tikzcd}\]
Squares are then used to model the term precision ordering in that
$\Delta \vdash M \ltdyn N : c$ implies $\sem{M}
\ltsq{\sem{\Delta}}{U\li\sem{c}} \sem{N}$ where here $\sem{\Delta =
  x_1:c_1,\ldots}$ is the action of the product on the relations
$c_1,\ldots$.


We note that for any $c$, there is a ``vertical identity'' square $\id
\ltsq{c}{c} \id$, as can be seen by unfolding the
definition. Similarly, for any monotone function $f : A_i \to A_o$,
there is a ``horizontal identity'' square $f
\ltsq{r(A_i)}{r(A_o)} f$ arising from the fact that $f$ is
monotone.
We can \emph{compose} squares vertically, i.e., if we have $f
\ltsq{c_1}{c_2} f'$ and $g \ltsq{c_2}{c_3} g'$ then we obtain a square
$g \circ f \ltsq{c_1}{c_3} g' \circ f'$. Likewise, we can compose
squares horizontally: If $f \ltsq{c_i}{c_o} g$ and $g
\ltsq{c_i'}{c_o'} h$, then we obtain a square $f \ltsq{c_i \comp
  c_i'}{c_o \comp c_o'} h$. Horizontal composition of squares
corresponds to transitivity of term precision. While we do not dwell
on categorical abstractions in this work, we note that this structure
of functions, relations and squares forms a locally thin \emph{double
category} as used in previous work \cite{new-licata18}.

Now that we have an intended model of term precision in the form of
squares, we can see what is required to give a compositional proof of
graduality. Most of the cases of term precision are just congruence
and follow easily. The core of the proof of the graduality property is
in proving the validity of the \emph{cast rules} $\upl$, $\upr$, $\dnl$, and
$\dnr$.
These rules specify a relationship between the semantics of type
precision derivations $c$ and the corresponding casts $\upc c, \dnc c$ (see Section
\ref{sec:GTLC}).
To prove graduality compositionally then, we will need to codify in an
extra property on $\sem{c}$ that the relation $\sem{c}$ has some
correspondence with the casts $\sem{\upc c}, \sem{\dnc c}$.
What property should we require?

Consider the $\upl$ rule. It states that if
$c : A_1 \ltdyn A_2$ and $c' : A_2 \ltdyn A_3$ and $M : A_1$ is
related to $N : A_3$ via the composition $c \comp c'$, then the upcast
$\upc{c} M$ is related to $N$ via $c'$. Since the upcasts are pure
functions in the semantics this can be deduced from the existence of
the following square, where the relation on the top is $\sem{c} \comp \sem{c'}$:
%
\[\begin{tikzcd}[ampersand replacement=\&]
  {\sem{A_1}} \& {\sem{A_2}} \& {\sem{A_3}} \\
  {\sem{A_2}} \&\& {\sem{A_3}}
  \arrow["\sem{c}", "\shortmid"{marking}, no head, from=1-1, to=1-2]
  \arrow["{\sem{\upc{c}}}"', from=1-1, to=2-1]
  \arrow["{\sem{c'}}", "\shortmid"{marking}, no head, from=1-2, to=1-3]
  \arrow["\id", from=1-3, to=2-3]
  \arrow["{\sem{c'}}"', "\shortmid"{marking}, no head, from=2-1, to=2-3]
\end{tikzcd}\]
While at first look this seems to be specifying a relationship between
$\sem{c}$ and $\sem{\upc c}$, there is a problem: it also quantifies
over an arbitrary other relation $\sem{c'}$! This means we cannot
require the existence of this square as part of the definition of a
relation between value types, as it is self-referential. This would
seem to imply that we cannot give a compositional model for
graduality. However, New and Licata observed that in the presence of
transitivity, we can \emph{derive} the above squares from simpler ones
that do not involve composition of relations. Below are the simpler squares
corresponding to $\upl$, $\upr$, $\dnl$, and $\dnr$:
\begin{center}
  \begin{tabular}{ c | c | c | c} 
    \begin{tikzcd}[ampersand replacement=\&]
      {A_1} \& {A_2} \\
      {A_2} \& {A_2}
      \arrow["c", "\shortmid"{marking}, no head, from=1-1, to=1-2]
      \arrow[""{name=0, anchor=center, inner sep=0}, "{\upc{c}}"', from=1-1, to=2-1]
      \arrow[""{name=1, anchor=center, inner sep=0}, "\id", from=1-2, to=2-2]
      \arrow["{r(A_2)}"', "\shortmid"{marking}, no head, from=2-1, to=2-2]
      \arrow["{\text{UpL}}"{marking, allow upside down}, draw=none, from=0, to=1]
  \end{tikzcd} &
    %
    \begin{tikzcd}[ampersand replacement=\&]
      {A_1} \& {A_1} \\
      {A_1} \& {A_2}
      \arrow["{r(A_1)}", "\shortmid"{marking}, no head, from=1-1, to=1-2]
      \arrow[""{name=0, anchor=center, inner sep=0}, "\id"', from=1-1, to=2-1]
      \arrow[""{name=1, anchor=center, inner sep=0}, "{\upc c}", from=1-2, to=2-2]
      \arrow["c"', "\shortmid"{marking}, no head, from=2-1, to=2-2]
      \arrow["{\text{UpR}}"{marking, allow upside down}, draw=none, from=0, to=1]
    \end{tikzcd} &
    %
    %
    %
    \begin{tikzcd}[ampersand replacement=\&]
      {A_2} \& {A_2} \\
      {A_1} \& {A_2}
      \arrow["{r(A_2)}", "\shortmid"{marking}, no head, from=1-1, to=1-2]
      \arrow[""{name=0, anchor=center, inner sep=0}, "{\dnc c}"', from=1-1, to=2-1]
      \arrow[""{name=1, anchor=center, inner sep=0}, "\id", from=1-2, to=2-2]
      \arrow["c"', "\shortmid"{marking}, no head, from=2-1, to=2-2]
      \arrow["{\text{DnL}}"{marking, allow upside down}, draw=none, from=0, to=1]
    \end{tikzcd} &
    %
    \begin{tikzcd}[ampersand replacement=\&]
      {A_1} \& {A_2} \\
      {A_1} \& {A_1}
      \arrow["c", "\shortmid"{marking}, no head, from=1-1, to=1-2]
      \arrow[""{name=0, anchor=center, inner sep=0}, "\id"', from=1-1, to=2-1]
      \arrow[""{name=1, anchor=center, inner sep=0}, "{\dnc c}", from=1-2, to=2-2]
      \arrow["{r(A_1)}"', "\shortmid"{marking}, no head, from=2-1, to=2-2]
      \arrow["{\text{DnR}}"{marking, allow upside down}, draw=none, from=0, to=1]
    \end{tikzcd}
  \end{tabular}
\end{center}
These squares say that the relation $c$ is \emph{representable} by the upcast
morphism $\upc{c}$, essentially that the relation is a kind of \emph{graph} of
the function in that $x \binrel{\sem c} y$ if and only if $\sem{\upc c}(x)
\leq_{A_2} y$ \cite{shulman2007}.

We can use the simpler $\upl$ square to derive the original $\upl$ square by
horizontally composing with the identity square for $c'$:
%
\[\begin{tikzcd}[ampersand replacement=\&]
	{A_1} \& {A_2} \& {A_3} \\
	{A_2} \& {A_2} \& {A_3}
	\arrow["c", "\shortmid"{marking}, no head, from=1-1, to=1-2]
	\arrow["{\upc{c}}"', from=1-1, to=2-1]
	\arrow["{c'}", "\shortmid"{marking}, no head, from=1-2, to=1-3]
	\arrow["\id"', from=1-2, to=2-2]
	\arrow["\id", from=1-3, to=2-3]
	\arrow["{\ltdyn_{A_2}}"', "\shortmid"{marking}, no head, from=2-1, to=2-2]
	\arrow["{c'}"', "\shortmid"{marking}, no head, from=2-2, to=2-3]
\end{tikzcd}\]
Then to complete the proof, we observe that the composition of $\ltdyn_{A_2}$ with $c'$ is equal to
$c'$, because $c'$ is \emph{downward-closed} under the relation on $A_2$.
An analogous argument shows that we can derive the original $\upr$ square from the simpler version of $\upr$,
this time using the fact that poset relations are \emph{upward-closed}.

To recap, we have shown that in order to carry out the graduality proof
compositionally, we need that the relations are closed under the ordering on
each side. In particular, this implies that the ordering relations $\ltdyn_{A}$
must be transitive, so we cannot entirely drop posets from our semantics.



\subsection{Resolution: Splitting Error Ordering and Bisimilarity}\label{sec:lock-step-and-weak-bisim}

Given that some amount of transitive reasoning is essential for
modeling type and term precision compositionally, we need to revisit
our ordering relation on $\li A$ to get one that is
transitive. Theorem \ref{thm:no-go} tells us, however, that any non-trivial such
relation must either not be a congruence with respect to $\theta$, or
must not be left- and right-step-insensitive.
\footnote{Technically, the theorem only implies that we need to drop
one of left- or right-step-insensitivity. We might then attempt
to work with two separate ``step-semi-sensitive'' relations on $\li
A$, each closed under delays on the left and right respectively. This
has some similarities to prior work on logical relations models
\cite{new-licata-ahmed2019,new-giovannini-licata-2022}, but we do not take this
approach because the relations still do not seem to be provably
transitive, though we have not been able to prove that they are not
transitive either.}
%
We cannot forego the property of being a $\theta$-congruence, as
without this we would not be able to prove basic properties of the
relation using \lob-induction, e.g., that the extension $f^\dagger$ is
\emph{monotone}, which is used pervasively in the proof of graduality.
Thus, we choose to sacrifice left and right step-insensitivity. That
is, we will define an ordering relation that requires terms to be in
``lock-step''. In order for two computations to be related in this
ordering, they must have the exact same stepping behavior
(unless/until the left-hand side results in an error).

More formally, we define the \emph{lock-step error-ordering}, with the idea
being that two computations $l$ and $l'$ are related if they are in lock-step
with regard to their intensional behavior, up to $l$ erroring. Figure
\ref{fig:lock-step-error-ordering} gives the definition of this relation.

\begin{figure}
  \begin{minipage}{0.5\textwidth}
    \fbox{$l_1 \ltls l_2$}
    \begin{align*}
        &\eta\, x \ltls \eta\, y \text{ if } 
            x \mathbin{\ltdyn_A} y \\
        &\mho \ltls l' \\
        &\theta\, \tilde{l} \ltls \theta\, \tilde{l'} \text{ if } 
            \later_t (\tilde{l}_t \ltls \tilde{l'}_t)
  \end{align*}\end{minipage}\begin{minipage}{0.5\textwidth}
\fbox{$l_1 \bisim l_2$}
    \begin{align*}
        \mho \bisim \mho &\text{ iff } \top \\
        \eta\, x \bisim \eta\, y &\text{ iff } x \bisim_A y \\
        \theta\, \tilde{x} \bisim \theta\, \tilde{y} &\text{ iff } \later_t (\tilde{x}_t \bisim \tilde{y}_t) \\
        \theta\, \tilde{x} \bisim \mho &\text{ iff } \exists n. \theta\, \tilde{x} = \delta^n(\mho)\\
        \theta\, \tilde{x} \bisim \eta\, y &\text{ iff } \exists n. \exists x \bisim_A y.
          (\theta\, \tilde{x} = \delta^n(\eta\, x))\\
        \mho \bisim \theta\, \tilde{y} &\text{ iff } \exists n. \theta\, \tilde{y} = \delta^n(\mho) \\
        \eta\, x \bisim \theta\, \tilde{y} &\text { iff } \exists n. \exists y \bisim_A x. (\theta\, \tilde{y} = \delta^n (\eta\, y))
      \end{align*}
  \end{minipage}
    \caption{lock-step error ordering and weak bisimilarity}
    \label{fig:lock-step-error-ordering}
\end{figure}

When both sides are $\eta$, then we check that the returned values are related
in $\ltdyn_A$. The error term $\mho$ is the least element. Lastly, if both sides step
(i.e., are a $\theta$) then we compare the resulting computations one time step
later.
It is straightforward to prove using \lob-induction that this relation is
reflexive, transitive and anti-symmetric given that the underlying relation $\ltdyn_A$
has those properties. The lock-step ordering is therefore the partial ordering
we will associate with $\li A$.
More generally we can define a heterogeneous version of this ordering that lifts
poset relation $c : A \rel A'$ to a poset relation $\li c : \li A \rel \li A'$.
The definition is nearly identical, except that in the case when both sides are $\eta$ we
check that the returned values are related in $c$.

However, we also cannot completely drop step-\emph{insensitivity} from
our model. The graduality property itself is not sensitive to steps,
and in fact there are terms related by term precision that take
differing numbers of steps.
The offending step arises precisely from the place where our
definition of $D$ uses a $\later$: the function case.
Consider the representability condition corresponding to the $\dnl$ rule
for $\iarr \colon \dyntodyn \ltdyn\, \dyn$. This is the square $\sem{\dnc{\iarr}} \ltsq{r({\li D})}{\li\sem{\iarr}} \id$.
For this square to be valid with the lock-step ordering the left and
right hand side would need to take the same number of steps, at least
when the left-hand side does not error.
However, the right hand side \emph{always} takes $0$ steps, as it is
the identity function, whereas on the left, our definition of
$\sem{\dnc{\iarr}}$ had to take an observable step when its input is a
function: this was inherent to the fact that the function case of the
dynamic type is guarded by a later.

Observe that we can remedy this particular situation by
replacing the $\id$ on the right hand side by an innocuous function
$\delta^* = (\delta \circ \eta)^\dagger$ that on an input value takes
a single computational step, but is otherwise the identity
function. If we were to ignore computational steps, then this function
\emph{would} be the identity function. We call such a function that is
the identity function except for the introduction of computational
steps a \emph{perturbation}.

We formalize this property of being equivalent ``except for steps''
with a second relation on the free error domain $\li A$: \emph{weak
bisimilarity}, defined in Figure~\ref{fig:lock-step-error-ordering}, which is
parameterized by a binary relation $\bisim_A$ on $A$
\cite{mogelberg-paviotti2016}.
Two errors are bisimilar, and when both sides are $\eta$, we ensure
that the underlying values are bisimilar in the underlying
bisimilarity relation on $A$. When both sides are thinking, we ensure
the terms are bisimilar later.  Most importantly, when one side is
thinking but the other terminates at $\eta x$ (i.e., one side steps),
we stipulate that the $\theta$-term runs to $\eta y$ where $x$ is
related to $y$. And similarly, if one side is thinking and the other
errors, we ensure the thinking side eventually errors.

It can be shown (by \lob-induction) that weak bisimilarity is
reflexive and symmetric. Since it is non-trivial, a $\theta$
congruence and step-insensitive on both sides, by
Theorem~\ref{thm:no-go}, we also know that it is \emph{not}
transitive.
We will then require our denotations of types to be not just posets,
but posets additionally equipped with a reflexive symmetric relation
$\bisim_A$.

Then we can refine our denotation of term precision $\Delta \vdash M
\ltdyn N : c$ to mean not that $\sem{M}$ and $\sem{N}$ are necessarily
in a lock-step error ordering directly, but that they can be
``synchronized'' to do so, that is that there exist $f \bisim \sem{M}$
and $g \bisim \sem{N}$ such that $f \ltsq{\sem{\Delta}}{U\li\sem{c}}
g$.
Note that this in turn implies that $\sem{M}$ and $\sem{N}$ are
related in the original \emph{step-insensitive} error ordering that we
sought to prove!
We additionally weaken our representability squares so that our
upcasts and downcasts are not required to be in a lock-step ordering
with an identity but instead to be in a lock-step ordering with some
perturbation, that is a function \emph{bisimilar to} the identity. We
call this weakened form of representability
\emph{quasi}-representability.
Then if our relations are quasi-representable we can in fact prove the
validity of our cast rules.
The remainder of the proof of graduality then is to show that this
property of being quasi-representable is itself compositional: that
all of the constructions we have on type precision preserve the
property of being quasi-representable.
While this is true for representability, it is not quite true for
quasi-representability. The reason is that perturbations are not quite
as well behaved as \emph{actual} identity functions. To solve this
final issue, we attach one final piece of information to our types
$A$: a type of \emph{syntactic} perturbations that represent functions
bisimilar to the identity, but are presented as data so that we can
perform operations such as composition and functorial actions on them.
With this notion of syntactic perturbation, and additional
requirements that the relations interact well with perturbations, we
finally achieve our desired result: a compositional,
syntax-independent proof of graduality.

\section{A Guarded Model of Graduality}\label{sec:concrete-relational-model}

Next we describe in detail our compositional model of GTLC that
satisfies graduality. The definitions proceed in two phases: in the
first phase we define our objects, functions, relations and squares to
model the error ordering and bisimilarity, and then we add the extra
structure of perturbations and quasi-representability to our objects
and relations to ensure that we can model the cast rules.
Once these definitions are complete we show how to interpret the
types, type precision, terms and term precision compositionally in our
model, and finally we conclude by extracting a big step semantics from
this model that satisfies the graduality theorem.

\subsection{Phase One: Predomains and Error Domains}

As discussed in the previous section our model of types must now come
with both an error ordering and a bisimilarity relation.
\begin{definition}
A \textbf{predomain} $A$ consists of a set $A$ along with two relations:
\begin{itemize}
    \item A partial order $\ltdyn_A$.
    \item A reflexive, symmetric ``bisimilarity'' relation $\bisim_A$.
\end{itemize}
A \textbf{morphism of predomains} $f : A \to A'$ is a function between
the underlying sets that is both \emph{monotone} (if
$x \ltdyn_A x'$, then $f(x) \ltdyn_{A'} f(x')$), and
\emph{bisimilarity-preserving} (if $x \bisim_A x'$, then $f(x) \bisim_{A'} f(x')$).
\end{definition}

We will generally write $A$ for both a predomain and its underlying set; if we
need to emphasize the difference we will write $|A|$ for the underlying set of
$A$.
Given a predomain $A$, we can form the predomain $\later A$. The
underlying set is $\later |A|$ and the ordering is given by $\tilde{x}
\ltdyn_{\later A} \tilde{x'}$ iff $\later_t(\tilde{x}_t \ltdyn_A
\tilde{x'}_t)$, and likewise for bisimilarity.
%

We similarly enhance our error domains with ordering and bisimilarity
relations, but we further need to impose conditions on the algebraic
structure: the error element should be the least element (so that it
models graduality), the think map should be monotone and bisimilarity-preserving,
and the bisimilarity relation should relate $\delta$ to the identity
(so that it models bisimilarity).
\begin{definition}
An \textbf{error domain} $B$ consists of a predomain $B$ along with the following data:
\begin{itemize}
    \item A distinguished ``error'' element $\mho_B \in B$ such that $\mho \ltdyn_B x$ for all $x$
    \item A morphism of predomains $\theta_B \colon \later B \to B$
    \item For all $x, y : B$ with $x \bisim_B y$, we have $\theta_B(\nxt\, x) \bisim_B y$
\end{itemize}
A \textbf{morphism of error domains} is a morphism of the underlying
predomains that preserves the error element and the $\theta$ map.
\end{definition}
For an error domain $B$, we define the predomain morphism $\delta_B := \theta_B
\circ \nxt$.

A \emph{predomain relation} $c : A \rel A'$ is a relation between $A$ and $A'$
that is upward closed in $A'$ and downward closed in $A$, as defined in the
previous section (no condition is required for bisimilarity). The ``reflexive''
predomain relation is written $r(A)$ and is given by the ordering $\ltdyn_A$.
The relational composition of two predomain relations is also a predomain
relation, and $r(A)$ is the identity for this composition.

Error domain relations are given by predomain relations that are
additionally a congruence with respect to the algebraic structures
$\theta$ and $\mho$.
\begin{definition}
  An \textbf{error domain relation} $d : B \rel B'$ is a monotone relation $d$ on the underlying predomains satisfying
  \begin{enumerate}
     \item ($\mho$ congruence): $\mho_B \binrel{d} \mho_{B'}$.
     \item ($\theta$ congruence): For all $\tilde{x}$ in $\laterhs B$ and
     $\tilde{y} \in \laterhs B'$, if $\laterhs_t (\tilde{x}_t \binrel{d} \tilde{y}_t)$, then 
     $(\theta_B (\tilde{x}) \binrel{d} \theta_{B'} (\tilde{y}))$. 
  \end{enumerate}
\end{definition}
Again we have a reflexive error domain relation $r(B)$ given by the
ordering. The relational composition of two error domain relations in
general is not an error domain relation ($\theta$ congruence in
particular is not preserved), but we instead use an alternative ``free
composition'' $dd' : B \rel B''$ that is inductively generated from the
relational composition. The rules are given in Figure \ref{fig:free-comp-ed-rel}.

\begin{figure}
  \begin{mathpar}
    \inferrule
    {x \mathbin{d} y \and y \mathbin{d'} z}
    {x \mathbin{(dd')} z}

    \inferrule
    {x' \ltdyn_{B_1} x \and x \mathbin{(dd')} z}
    {x' \mathbin{(dd')} z}

    \inferrule
    {x \mathbin{(dd')} z \and z \ltdyn_{B_3} z'}
    {x \mathbin{(dd')} z'}

    \inferrule
    { }
    {\mho_{B_1} \mathbin{(dd')} z}

    \inferrule
    {\later_t[ \tilde{x}_t \mathbin{(dd')} \tilde{z}_t ] }
    {(\theta_{B_1}\, \tilde{x}) \mathbin{(dd')} (\theta_{B_3}\, \tilde{z}) }
  \end{mathpar}
  \caption{The free composition of error domain relations.}
  \label{fig:free-comp-ed-rel}
\end{figure}

As in the previous section, we also have the notion of a square $f
\ltsq{c_i}{c_o} g$ between predomain morphisms, which is defined to mean that
for all $x \binrel{c_i} y$, we have $f(x) \binrel{c_o} g(y)$. Likewise, we have
squares $\phi \ltsq{d_i}{d_o} \psi$ between error domain morphisms, which are
defined similarly, with no extra interaction with the algebraic structure.
The free composition of error domain relations has the universal property that
$\phi \ltsq{dd'}{d''} \phi'$ if and only if for any $x \binrel{d} y$
and $y \binrel{d'} z$, we have $\phi(x) \binrel{d''} \phi'(z)$.

In our denotational model of GTLC, we map every type $A$ to a
predomain $\sem{A}$. These constructions are given by the same
constructions as in Section~\ref{sec:concrete-term-model}, but now enriched to
act also on orderings and bisimilarity. For products this is the
obvious lifting, and for functions we use the CBPV decomposition $U(A
\to \li A')$ where the only one that is not clear is $\li A'$.
\begin{definition}
  The free error domain $\li A$ has as the underlying set the free
  simple error domain with ordering given by the lock-step error
  ordering and bisimilarity by weak bisimilarity both generated by the
  ordering/bisimilarity on $A$. This is free in the same sense as the
  free simple error domain that it is left adjoint to $U$.
\end{definition}

Additionally, all of these core constructions $\times, \to, U, \li$
act on heterogeneous relations as well in an analogous way.

The most interesting type to interpret is the dynamic
type. The beauty of guarded domain theory is that once we have define
$\laterhs$ for predomains, we can solve the same domain equation as in
Section~\ref{sec:concrete-term-model} changed only in that we interpret it as a
domain equation on predomains rather than sets:
\[ D \cong \mathbb{N}\, + (D \times D)\, + \laterhs U(D \arr \li D) \]
%
The complete definition of the ordering $\ltdyn_D$ is provided in the appendix
for completeness.
%
%
We refer to the three injections as $e_\nat$, $e_\times$,
and $e_\to$, respectively.

To interpret type precision, we define three relations involving $D$
as follows. We define $\inat : \mathbb{N} \rel D$ by $(n, d) \in
\inat$ iff $e_\nat(n) \ltdyn_D d$. We similarly define $\itimes : D
\times D \rel D$ by $((d_1, d_2), d) \in \itimes$ iff $e_\times(d_1,
d_2) \ltdyn_D d$, and we define $\text{inj}_\to : U(D \arr \li D) \rel
D$ by $(f, d) \in \iarr$ iff $(e_\to(f) \circ \nxt) \ltdyn_D d$.

\subsection{Phase Two: Perturbations and Quasi-representable Relations}

We now introduce the additional definitions needed to complete the construction
of our model. Recall from Section \ref{sec:lock-step-and-weak-bisim} that in order
to establish the $\dnl$ rule for the downcast of $\iarr$, we needed to adjust the
rule by inserting a ``delay'' on the right-hand side. A similar adjustment is
needed for the $\dnr$ rule for $\iarr$. Moreover, the need to insert delays impacts
the semantics of the cast rules for \emph{all} relations, because of the
functorial nature of casts. That is, the upcast at a function type $c_i \ra c_o$
involves a downcast in the domain and an upcast in the codomain. This has two
consequences: first, the squares corresponding to the rules $\upl$ and $\upr$ may also
require the insertion of a delay. Second, we need to be able to insert
``higher-order'' delays in a way that follows the structures of the casts.

To accomplish this, we equip the predomains $A$ in our semantics with
a monoid $M_A$ of \emph{syntactic perturbations}, as well as a means
of \emph{interpreting} these syntactic perturbations as
\emph{semantic} perturbations on $A$, i.e., as endomorphisms that are
bisimilar to the identity. We call such an object a \emph{value
object}, and this will serve as the final definition of a denotation of
a type in our semantics. Similarly, we define computation objects to be
error domains equipped with a notion of syntactic perturbation:
\begin{definition}
  A \textbf{value object} consists of a predomain $A$, a monoid $M_A$ and a
  homomorphism of monoids $i_A \colon M_A \to \morbisimid{A} := \{ f : A \to A \mid f
  \bisim \id_A \}$.

  Likewise, a \textbf{computation object} consists of an error domain $B$, a
  monoid $M_B$ and a homomorphism of monoids $i_B \colon M_B \to \morbisimid{B}$.

  We define a \textbf{value morphism} to be simply a morphism of the underlying
  predomains, and a \textbf{computation morphism} to be a morphism of the
  underlying error domains.
\end{definition}

\begin{remark}
  Syntactic perturbations are equipped with a monoid structure in
  order to support identity and composition relations on value and
  computation types. That is, the identity relation requires an
  identity perturbation and composition of relations requires a
  composition operation on perturbations.
\end{remark}

\begin{remark}
  It might seem strange that we equip objects with a monoid of
  perturbations but we do not require morphisms to respect this
  structure in any way. A natural attempt would be defining a morphism
  between value objects $(A, M_A, i_A)$ and $(A', M_{A'}, i_{A'})$ to
  be a predomain morphism $f \colon A \to A'$ along with a monoid
  homomorphism $h \colon M_A \to M_{A'}$ such that $f$ and $h$ satisfy
  a compatibility condition with respect to the interpretations $i_A$
  and $i_{A'}$. However, this condition simply fails to hold for all
  terms, in particular the property is preserved by $\lambda$
  abstraction.
  %
  
\end{remark} 


Having defined syntactic perturbations and value and computation objects, we can
now make formal the final aspect of the model: the notion of
quasi-representability of relations.
\begin{definition}\label{def:quasi-left-representable}
Let $(A, M_A, i_A)$ and $(A', M_{A'}, i_{A'})$ be value objects, and let $c : A
\rel A'$ be a predomain relation. We say that $c$ is
\textbf{quasi-left-representable} by a predomain morphism $e_c : A \to A'$ if there
are perturbations $\delle_c \in M_A$ and $\delre_c \in M_{A'}$ such that the
following two squares exist: 
(1) $\upl$: $e_c \ltsq{c}{r(A')} i_{A'}(\delre_c)$, and
(2) $\upr$: $i_A(\delle_c) \ltsq{r(A)}{c} e_c$.

Quasi-left-representability for error domain relations is defined analogously.
\end{definition}
Observe that this weakens the previous notion of representability, since under
that definition the perturbations were required to be the identity.
%

Likewise, we define quasi-right-representability as follows:
\begin{definition}\label{def:quasi-right-representable}
Let $(B, M_B, i_B)$ and $(B', M_{B'}, i_{B'})$ be computation objects, and let
$d : B \rel B'$ be an error domain relation. We say that $d$ is
\textbf{quasi-right-representable} by an error domain morphism $p_d : B' \multimap B$ if
there are perturbations $\dellp_d \in M_B$ and $\delrp_d \in M_{B'}$ such that
the following two squares exist: 
(1) $\dnl$: $p_d \ltsq{r(B')}{d} i_{B'}(\delrp_d)$, and
(2) $\dnr$: $i_B(\dellp_d) \ltsq{d}{r(B)} p_d$.

Quasi-right-representability for predomain relations is defined similarly.
\end{definition}

For ordinary representability, the relation is \emph{uniquely
determined} by any morphism that left- or right-represents it: if two
relations are left-represented by the same morphism, then they are
equal. But for quasi-representability, this is no longer the case due
to the presence of perturbations. Instead we have the following weaker property:
\begin{definition}\label{def:quasi-equivalent}
  Let $A$ and $A'$ be value objects, and let $c, c' : A \rel A'$ be relations
  between the underlying predomains. We say that $c$ and $c'$ are
  \textbf{quasi-equivalent}, written $c \qordeq c'$, if there exist
  syntactic perturbations $\delta^l_1, \delta^l_2 \in M_A$ and $\delta^r_1,
  \delta^r_2 \in M_{A'}$ such that there is a square $i_A(\delta^l_1)
  \ltsq{c}{c'} i_{A'}(\delta^r_1)$ and a square $i_A(\delta^l_2) \ltsq{c'}{c}
  i_{A'}(\delta^r_2)$. Given computation objects $B$ and $B$, we make the
  analogous definition for relations $d, d' : B \rel B'$ between the underlying
  error domains.
\end{definition}

Given two quasi-representable relations $c$ and $c'$, it is not the case in
general that $cc'$ is quasi-representable. We need an additional condition that
specifies how relations interact with perturbations. In particular, we need to
be able to ``push'' and ``pull'' perturbations along relations $c$ and $d$.
The intuition for this requirement comes from the construction of the square
corresponding to the original $\upl$ rule from the square for the simplified
version as was shown in Section \ref{sec:towards-relational-model}. Recall that
we horizontally composed the simplified $\upl$ square for $c$ with the identity
square for $c'$ on the right. Now that the square for $\upl$ involves a
perturbation on the right instead of the identity, this construction will not
work unless we can ``push'' the perturbation on $A_2$ to one on $A_3$. That is,
given a relation $c : A \rel A'$ and a syntactic perturbation $m_A \in M_A$, we
need to be able to turn it into a syntactic perturbation $m_{A'} \in M_{A'}$
such that the resulting semantic perturbations obtained by applying the
respective homomorphisms $i_A$ and $i_{A'}$ form a square. Dually, we will need
to ``pull'' a perturbation from right to left. This notion is made formal by the
following definition:
\begin{definition}
    Let $(A, M_A, i_A)$ and $(A', M_{A'}, i_{A'})$ be value objects,
    and let $c : A \rel A'$ be a predomain relation. A
    \textbf{push-pull structure} for $c$ (with respect to $M_A$ and
    $M_{A'}$) consists of monoid homomorphisms $\push : M_A \to M_A'$
    and $\pull : M_A' \to M_A$ such that for all $m_A \in M_A$ there
    is a square $i_A(m_A) \ltsq{c}{c} i_{A'}(\push\, m_A)$ and for all
    $m_{A'} \in M_{A'}$ there is a square $i_A(\pull\, m_{A'})
    \ltsq{c}{c} i_{A'}(m_{A'})$. A push-pull structure for error
    domain relations is defined similarly.
\end{definition}

We can now give the final definition of value and computation relations for our
model, which serve as the denotation of our type precision derivations:
\begin{definition}
A \textbf{value relation} between value objects $(A, M_A, i_A)$ and $(A',
M_{A'}, i_{A'})$ consists of a predomain relation $c : A \rel A'$ that has a
push-pull structure, is quasi-left-representable by a morphism $e_c : A \to A'$,
and is such that $\li c$ is quasi-right-representable by an error domain morphism
$p_c : \li A' \to \li A$.

A \textbf{computation relation} between computation objects $(B, M_B, i_B)$ and
$(B', M_{B'}, i_{B'})$ consists of an error domain relation $d$ that has a
push-pull structure, is quasi-right-representable by a morphism $p_d : B' \to
B$, and is such that $Ud$ is quasi-left-representable by a morphism $e_d : UB
\to UB'$
\end{definition}

It is essential in the definition of a value relation that we have not
just that $c$ is quasi-left-representable but also that $\li c$ is
quasi-right-representable so that we can define the action of the
functor $\li$ on relations, taking a value relation to a computation
relation. Likewise, we require $Ud$ to be quasi-left-representable in
the definition of computation relations so that we can define the
action of $U$ on relations.


To model the term precision in a way that is oblivious to the stepping
behavior of terms, we need to weaken the notion of square to allow for the
morphisms on either side to be ``out of sync'' rather than in lock step.
In the below, $f' \bisimsq{A_i}{A_o} f$ is the natural extension of
bisimilarity to morphisms: given bisimilar inputs, the functions have
bisimilar outputs.
\begin{definition}
  Let $c_i : A_i \rel A_i'$, $c_o : A_o \rel A_o'$ $f : A_i \to A_o$
  and $g : A_i' \to A_o'$.  An \textbf{extensional square} $f
  \ltsqbisim{c_i}{c_o} g$ holds if there \emph{exist} $f'
  \bisimsq{A_i}{A_o} f$ and $g' \bisimsq{A_i'}{A_o'} g$ such that $f' \ltsq{c_i}{c_o} g'$.
\end{definition}

\subsection{Constructing the Denotational Semantics}\label{sec:relational-model-construction}

With these concepts now defined, we can give an overview of our
denotational model:
\begin{itemize}
  \item For every type $A$ we define a value type $\sem{A}$, and define $\sem{x_1:A_1,\ldots} = \sem{A_1}\times\cdots$
  \item For every term $\Gamma \vdash M : A$, we define a value morphism $\sem{M} : \sem{\Gamma} \to U\li\sem{A}$.
  \item If $M = M'$ are equal according to our equational theory, then $\sem{M} = \sem{M'}$
  \item For every type precision derivation $c : A \ltdyn A'$ we
    define a value relation $\sem{c} : \sem{A} \rel \sem{A'}$.
  \item For every type precision equivalence $c \equiv c'$ we have a
    quasi-equivalence of relations $\sem{c} \bisim \sem{c'}$.
  \item For every term precision $\Delta \vdash M \ltdyn M' : c$ we
    have an \emph{extensional} square
    $\sem{M}\ltsqbisim{\sem{\Delta}}{U\li\sem{c}} \sem{M'}$
\end{itemize}

The term semantics $\sem{M}$ is nearly identical to the term semantics in
Section~\ref{sec:concrete-term-model}. One difference is that the interpretation
of the function type $A_i \ra A_o$ in the term semantics of Section
\ref{sec:concrete-term-model} involves \emph{arbitrary functions}, while the
interpretation of the function type in the term semantics defined here involves
only those functions that are morphisms of predomains. Additionally, the
semantics of casts is slightly different. In the semantics of Section
\ref{sec:concrete-term-model}, the projection corresponding to the type
precision $cc'$ is simply the composition of the projections. In the new
semantics, the corresponding projection involves a perturbation.  
%
This is a consequence of the construction of the composition of
value/computation relations (the relevant projection is defined in the proof
of Lemma \ref{lem:representation-comp-F-U} in the Appendix).
%
%
However, the two projections are extensionally equivalent since they are the
same up to a perturbation.


The validity of the equational theory follows easily: $\beta\eta$ holds in our
model. In the remainder of this section, we elaborate on each of the remaining
components of the model.

\subsubsection{Interpreting Types}\label{sec:perturbation-constructions}

For our value type semantics, we already have an interpretation as predomains;
we describe how we enhance each of the constructions we use ($U,\li,\to,\times,D$)
with a monoid of syntactic perturbation and a means of interpreting them as semantic
perturbations.

First consider the free error domain $\li-$. Given a value object $(A, M_A,
i_A)$ we define the syntactic perturbations $M_{\li A}$ to be $\mathbb{N} \oplus
M_A$, where $\oplus$ denotes the free product of monoids (the coproduct in the
category of monoids). The intuition behind needing to include $\mathbb{N}$ comes
from the example of the downcast for $\iarr$. Recall that we adjusted the
squares corresponding to the $\dnl$ and $\dnr$ rules by adding $(\delta \circ
\eta)^\dagger$ on the side opposite the downcast. Thus, we need a syntactic
perturbation in $M_{\li A}$ that will be interpreted as $(\delta \circ
\eta)^\dagger$. We accomplish this by taking the coproduct with $\mathbb{N}$.
Since $\mathbb{N}$ is the free monoid on one generator, defining a homomorphism
from $\mathbb{N}$ to $\morbisimid{\li A}$ is equivalent to a choice of an
element in $\morbisimid{\li A}$.
The interpretation of the perturbations as endomorphisms is defined using the
universal property of the coproduct of monoids. Specifically, to define the
homomorphism $i_{\li A}$, by the universal property of the coproduct of monoids
it suffices to define two homomorphisms, one from $\mathbb{N} \to
\morbisimid{\li A}$ and one from $M_A \to \morbisimid{\li A}$. The first
homomorphism sends the generator $1 \in \mathbb{N}$ to the error domain
endomorphism $(\delta \circ \eta)^\dagger$, and the second sends $m_A \in M_A$ to
$\li(i_A(m_A))$. We observe in both cases that the resulting endomorphisms are
bisimilar to the identity, since in the first case, $-^\dagger$ preserves
bisimilarity, and in the second case, the action of $\li$ on morphisms preserves
bisimilarity.

Now we define the action of $U$ on objects. Given a computation object $(B, M_B,
i_B)$, we define $M_{UB}$ to be $\mathbb{N} \oplus M_B$. The reason for
requiring $\mathbb{N}$ is related to the Kleisli action of the arrow functor: to
establish quasi-representability of $U(c \arr d)$ given quasi-representability
of $c$ and $d$, we will need to turn a perturbation on $\li A$ into a
perturbation on $U(A \arr B)$. Since the perturbations for $\li A$ involve
$\mathbb{N}$, so must the perturbations for $UB$. The interpretation $i_{UB}$ of
the perturbations on $UB$ works in an analogous manner to that of $\li A$: we send
the generator $1 \in \mathbb{N}$ to the delay morphism $\delta_B \colon UB \to
UB$, and we send a perturbation $m_B \in M_B$ to $U(i_B(m_B))$.

The action of $\times$ on objects is as follows. Given $(A_1, M_{A_1}, i_{A_1})$
and $(A_2, M_{A_2}, i_{A_2})$, we define the monoid $M_{A_1 \times A_2} =
M_{A_1} \oplus M_{A_2}$. The interpretation $i_{A_1 \times A_2}$ is defined
using the universal property of the coproduct of monoids. That is, we define
homomorphisms $M_{A_1} \to \morbisimid{A_1 \times A_2}$ and $M_{A_2} \to
\morbisimid{A_1 \times A_2}$. The former takes $m_1 \in M_{A_1}$ to the morphism
$(i_{A_1}(m_1), \id)$, while the latter takes $m_2 \in M_{A_2}$ to the morphism
$(\id, i_{A_2}(m_2))$.

The action of $\arr$ on objects is as follows. Given $(A, M_A, iA)$ and $(B,
M_B, i_B)$ we define the monoid $M_{A \arr B} = M_A^{op} \oplus M_B$. (Here,
$M^{op}$ denotes the monoid with the same elements but with the order of
multiplication reversed, i.e., $x \cdot_{M^{op}} y := y \cdot_M x$.) To define
the interpretation $i_{A \arr B}$, it suffices to define a homomorphism
$M_A^{op} \to \morbisimid{A \arr B}$ and $M_B \to \morbisimid{A \arr B}$. The
former takes an element $m_A$ and a predomain morphism $f : A \to UB$ and
returns the composition $f \circ i_A(m_A)$, while the latter is defined
similarly by post-composition with $i_B$.

Finally, the most interesting type is of course the dynamic type. 
It may seem as though we need to define the perturbations for $D$ by
guarded recursion, but in fact we define them via least-fixpoint
in the category of monoids:
\( M_D \cong (M_{D \times D}) \oplus M_{U(D \to \li D)}. \)
\begin{remark}
  Here it is crucial to define the perturbations of the dynamic type
  \emph{inductively} rather than by guarded recursion. If we were to
  define the perturbations by guarded recursion, then it would be
  impossible to define the pull homomorphism for the $\iarr$ relation.
\end{remark}
We now explain how to interpret these perturbations as endomorphisms.
We define $i_D : M_D \to \{ f : D \to D \mid f \bisim \id \}$ by induction as follows:
\begin{align*}
  i_D(m_\times) &= \lambda d. \text{case $d$ of }
    \{ \itimes (d_1, d_2) \To \itimes (i_{D \times D}(m_\times)(d_1, d_2)) \alt x \To x \} \\
  i_D(m_\to) &= \lambda d. \text{case $d$ of }
    \{ \iarr (\tilde{f}) \To \iarr (\lambda t. i_{U(D \to \li D)}(m_\to)(\tilde{f}_t)) \alt x \To x \}
\end{align*}
In the case of a perturbation on $D \times D$, we use the interpretation for the
perturbations of $D \times D$ as defined earlier in this section, which in turn
will use $i_D$ inductively. Similarly, for a perturbation on $U(D \to \li D)$ we
use the interpretation for perturbations on $U(D \to \li D)$.
One can verify that this defines a homomorphism from $M_D \to \{ f : D \to D : f
\bisim \id \}$.


\subsubsection{Interpreting Type Precision}

Next we need to confirm that our interpretation of type precision
derivations as predomain relations indeed extends to value relations,
i.e., that the relations are quasi-representable and that they satisfy
the push-pull property for perturbations.
%
We include most of these technical verifications in the appendix (Definitions
\ref{def:value-computation-rel-comp} and
\ref{def:functorial-actions-on-relations}).
%

We go into more detail on the three relations involving the dynamic
type $\inat$, $\itimes$, and $\iarr$. As an illustrative case, we
establish the push-pull property for the relation $\iarr$. We define
$\pull_{\iarr} : M_D \to M_{U(D \arr \li D)}$ by giving a homomorphism
from $M_{D \times D} \to M_{U(D \arr \li D)}$ and from $M_{U(D \to \li
  D)} \to M_{U(D \arr \li D)}$. The former is the trivial homomorphism
sending everything to the identity element, while the latter is the
identity homomorphism.
Conversely, we define $\push_{\iarr} : M_{U(D \arr \li D)} \to M_D$ as the
inclusion into the coproduct.
Showing that the relevant squares exist involving push and pull is
straightforward.

It is easy to see that the relations $\inat$, $\itimes$, and $\iarr$
are quasi-left-representable, where all perturbations are taken to be the
identity elements of the respective monoids.
For quasi-right-representability, the only nontrivial case is $\li(\iarr)$.
Recall the definition of the downcast for $\iarr$ given in Section
\ref{sec:term-interpretation}. We know that in order for the $\dnr$ and $\dnl$ squares
for $\li \iarr$ to exist, we must insert a delay on the side opposite the
downcast. In terms of syntactic perturbations, this means that we take $\dellp$
and $\delrp$ in the definition of quasi-right-representability to both be
$\inl\, 1$, where $\inl$ is the left injection into the coproduct of monoids. We
recall that the definition of syntactic perturbations for $\li$ involves a
coproduct with $\mathbb{N}$, and that the interpretation homomorphism maps the
generator $1$ to the endomorphism of error domains $(\delta \circ
\eta)^\dagger$.
With this choice for the perturbations, it is straightforward to show that the
$\dnl$ and $\dnr$ squares exist using the definition of the downcast and the
definition of the relation $\iarr$.

Lastly, the interpretation of the type precision equivalence rules $c \equiv c'$
as quasi-equivalence of the denoted relations $\sem{c} \bisim \sem{c'}$ is
justified by Lemma \ref{lem:quasi-order-equiv-functors} in the Appendix.
%

\subsubsection{Interpreting Term Precision: Extensional Squares}

Finally, we establish the core result of the semantics: interpreting
term precision as extensional squares. Fortunately, the hard work is
now complete, and the term precision rules follow from our
compositional constructions.
In particular, the congruence rules of term precision are modeled easily: congruence
squares exist both for bisimilarity and for the error ordering, and so
we also have these for extensional squares.
The \textsc{ErrBot} rule holds because $\mho$ is the least element in the semantics of the . 
Next, the \textsc{EquivTyPrec} rule follows because $c \equiv c'$
implies quasi-equivalence of the denoted relations, and this in turn implies that
$\id \ltsqbisim{\sem{c}}{\sem{c'}} \id$.
Lastly, it remains to verify the existence of the squares for the cast rules. For $\upl$, the
relevant extensional square is obtained as follows:
%
\[\begin{tikzcd}[ampersand replacement=\&,column sep=large]
	{A_1} \& {A_2} \& {A_3} \\
	{A_2} \& {A_2} \& {A_3}
	\arrow["c", "\shortmid"{marking}, no head, from=1-1, to=1-2]
	\arrow[""{name=0, anchor=center, inner sep=0}, "{e_c}"', curve={height=6pt}, from=1-1, to=2-1]
	\arrow[""{name=1, anchor=center, inner sep=0}, "{e_c}", curve={height=-6pt}, from=1-1, to=2-1]
	\arrow["{c'}", "\shortmid"{marking}, no head, from=1-2, to=1-3]
	\arrow["{\delre_c}"', from=1-2, to=2-2]
	\arrow[""{name=2, anchor=center, inner sep=0}, "{\push_{c'}(\delre_c)}"', curve={height=6pt}, from=1-3, to=2-3]
	\arrow[""{name=3, anchor=center, inner sep=0}, "\id", curve={height=-6pt}, from=1-3, to=2-3]
	\arrow["{r(A_2)}"', "\shortmid"{marking}, no head, from=2-1, to=2-2]
	\arrow["{c'}"', "\shortmid"{marking}, no head, from=2-2, to=2-3]
	\arrow["\bisim"{description}, draw=none, from=0, to=1]
	\arrow["\bisim"{description}, draw=none, from=2, to=3]
\end{tikzcd}\]
%
We have used the push-pull structure on $c'$ to turn the perturbation on $A_2$
into a perturbation on $A_3$. Lastly we note that the bottom edge of the square
is equal to $c'$ by the fact that $c'$ is downward closed. The squares for $\upr$,
$\dnl$, and $\dnr$ are obtained analogously.


\subsection{Relational Adequacy}\label{sec:adequacy}

Now that we have established our denotational model, we can, as in
Section~\ref{sec:big-step-term-semantics}, extract a big step semantics. Our final
goal then is to prove that the extensional squares in our model imply
the graduality property for this final big step semantics.
The result we seek to show is the following:
\begin{theorem}[adequacy]
  If $\cdot \vdash M \ltdyn N : \nat$, then:
  \begin{itemize}
    \item If $M \Downarrow n$ then $N \Downarrow n$.
    \item If $N \Downarrow \mho$ then $M \Downarrow \mho$.
    \item If $N \Downarrow n$ then $M \Downarrow n$ or $M \Downarrow \mho$.
  \end{itemize}
\end{theorem}
%
%
It therefore remains to show that the denotation of term precision in our model
$f \ltsqbisim{}{} g$, where $f, g : \li \mathbb{N}$, implies a relationship
between the partial elements $\mathbb{N} + {\mho}$ corresponding to $f$ and $g$. Recall
that to define the big-step semantics we used the equivalence of the
globalization of the free error domain and Capretta's coinductive delay monad:
$\li^{gl} \mathbb{N} \cong \delay(\mathbb{N} + {\mho})$. We now extend this
correspondence to a relationship between a global version of our denotation of
term precision and the corresponding relation on the delay monad.


To that end, we define, using clock quantification, a version of our denotation of term precision
that relates terms of type $\li^{gl} \mathbb{N}$. More concretely, consider the
predomain $\mathbb{N}$ with ordering and bisimilarity both the equality
relation. We define a global version of the lock-step error ordering and the
weak bisimilarity relation on elements of $\li^{gl} \mathbb{N}$; the former is
defined by
\( x \ltls^{gl} y := \forall k. x[k] \ltls y[k], \)
and the latter is defined by
\( x \bisim^{gl} y := \forall k. x[k] \bisim y[k]. \)
In these two definitions, we apply the clock $k$ to the elements $x$ and $y$ of
type $\li^{gl} \mathbb{N} := \forall (k : \Clock). \li^k \mathbb{N}$.

Correspondingly, we define a version of the lock-step error ordering and weak
bisimilarity relation on $\delay(\mathbb{N} + {\mho})$ (see Section
\ref{sec:relations-on-delay-monad} in the appendix for the definitions).
%
%
%
%
By adapting Theorem 4.3 of Kristensen et al. \cite{kristensen-mogelberg-vezzosi2022} to the
setting of inductively-defined relations, we can show that both the global
lock-step error ordering and the global weak bisimilarity admit coinductive
definitions. In particular, modulo the above isomorphism between $\li^{gl} X$
and $\delay(\mathbb{N} + {\mho})$, the global version of the lock-step error
ordering is equivalent to the lock-step error ordering on $\delay(\mathbb{N} +
{\mho})$, and the global version of weak bisimilarity is equivalent to weak
bisimilarity on $\delay(\mathbb{N} + {\mho})$.
This implies that the global denotation of the term precision semantics for
$\li^{gl} \mathbb{N}$ agrees with the corresponding relation $\ltsqbisim{}{}^\text{Delay}$ for $\delay(\mathbb{N} + {\mho})$, i.e.,
the closure of its lock-step ordering under weak bisimilarity on both sides.
Then to obtain the desired result, it suffices to show that being related in
$\ltsqbisim{}{}^\text{Delay}$ implies that the corresponding partial elements of
$\mathbb{N} + {\mho}$ are related in the expected manner. Specifically, recall
from Section \ref{sec:big-step-term-semantics} that $-\Downarrow^\text{Delay}$
is a partial function from delays to $\mathbb{N} + {\mho}$ defined by
existential quantification over the number of steps in which the argument terminates.
We want to show that if $d \ltsqbisim{}{}^\text{Delay} d'$, then 
$(d \Downarrow^\text{Delay}) \ltdyn^\text{Partial}\, (d' \Downarrow^\text{Delay})$,
where the relation $\ltdyn^\text{Partial}$ is defined such that $\mho$ is the
least element, $\bot \ltdyn^\text{Partial} \bot$, and $n \ltdyn^\text{Partial}
n$. This implication follows from the definition of
$\ltsqbisim{}{}^\text{Delay}$.





\section{Discussion}\label{sec:discussion}

\subsection{Related Work}

We give an overview of current approaches to proving graduality and discuss
their limitations where appropriate.


The methods of Abstracting Gradual Typing \cite{garcia-clark-tanter2016} and the
formal tools of the Gradualizer \cite{cimini-siek2016} provide a way to derive
from a statically-typed language a language that satisfies graduality by
construction. The main downside to these approaches lies in their inflexibility:
since the process in entirely mechanical, the language designer must adhere to
the predefined framework.  Many gradually typed languages do not fit into either
framework, e.g., Typed Racket \cite{tobin-hochstadt06, tobin-hochstadt08}, and
the semantics produced is not always the desired one.
Furthermore, while these frameworks do prove graduality of the resulting
languages, they do not show the correctness of the equational theory, which is
equally important to sound gradual typing.


A line of work by New, Licata and Ahmed
\cite{new-ahmed2018,new-licata18,new-licata-ahmed2019} develops an axiomatic
account of gradual typing and proves graduality by interpreting type precision
$c : A \ltdyn A'$ as an \emph{embedding-projection pair}, that is, a pure
embedding function $e : A \to A'$ and a possibly erroring projection $p : \li A'
\multimap \li A$ such that $p \circ e = \id$ and $e \circ p \ltdyn \id$. At
first glance, our model of precision looks different: it is a \emph{relation}
$\sem{c} : A \rel A'$ with a quasi-representability structure that gives us
something close to $e$ and $p$ in an embedding-projection pair. The relationship
between these models is that \emph{if} the relation were \emph{truly}
representable we would have that $e$ and $p$ form a \emph{Galois connection} $e
\circ p \ltdyn \id$ and $\id \ltdyn p \circ e$. However, we have dropped the
stronger property of retraction from our analysis in this work. With true
representability, the relation $c$ is \emph{uniquely determined} by the
embedding $e$, but since we only have quasi-representability, we need to keep
the relation around explicitly. So the extra complexity of managing explicit
relations is a cost of the intensional reasoning that guarded type theory
introduces. 
The line of work by New, Licata, and Ahmed involved both denotational methods
and operational logical relations approaches.




On the denotational side, New and Licata \cite{new-licata18} showed that the
graduality proof could be modeled using semantics in certain kinds of double
categories, which are categories internal to the category of categories. A
double category extends a category with a second notion of morphism, often a
notion of ``relation'' to be paired with the notion of functional morphism, as
well as a notion of functional morphisms preserving relations. In gradual
typing, the notion of relation models type precision and the squares model the
term precision relation.
%
%
The double-categorical framework serves as a foundation on which we base our
model construction in this paper. However, our model is not actually a double
category: the lack of transitivity of bisimilarity means that we cannot compose
extensional squares horizontally.

With the notion of abstract categorical model for gradual typing in hand,
denotational semantics is then the work of constructing concrete models that
exhibit the categorical construction. New and Licata \cite{new-licata18} present
such a model using categories of $\omega$-CPOs, and this model was extended by
Lennon-Bertrand, Maillard, Tabareau and Tanter \cite{gradualizing-cic} to prove
graduality of a gradual dependently typed calculus $\textrm{CastCIC}^{\mathcal
G}$. This domain-theoretic approach meets our criterion of being a semantic
framework for proving graduality, but suffers from the limitations of classical
domain theory: the inability to model viciously self-referential structures such
as higher-order extensible state and similar features such as runtime-extensible
dynamic types.



On the operational side, a series of works
\cite{new-ahmed2018, new-licata-ahmed2019, new-jamner-ahmed19}
developed step-indexed logical relations models of gradually typed languages.
Unlike classical domain theory, such step-indexed techniques can scale to
essentially arbitrary self-referential definitions, which means they can model
higher-order store and runtime-extensible dynamic types
\cite{appelmcallester01,ahmed06,neis09,new-jamner-ahmed19}. However, as is
common with operational approaches, their proof developments are highly
repetitive and technical, with each development formulating a logical relation
from first-principles and proving many of the same tedious lemmas without
reusable mathematical abstractions.
%
%
%
This is addressed somewhat by Siek and Chen \cite{siek-chen2021}, who give a
proof in Agda of graduality for an operational semantics using a \emph{guarded
logic of propositions} shallowly embedded in Agda. The guarded logic simplifies
the treatment of step-indexed logical relations, but the approach is still
fundamentally operational, and so the main lemmas of the work are still tied to
the particular operational syntactic calculus being modeled.

Our work combines the benefits of both the denotational and step-indexed
approaches. We take as a starting point the denotational approach of New and
Licata, but we work with guarded type theory (a synthetic form of step-indexing)
rather than classical domain theory. As we have seen, in the guarded setting the compositional
theory of New-Licata breaks down. 
Our work provides a novel refinement to their theory that allows for some
compositional reasoning to be recovered even in the guarded setting.

Eremondi \cite{Eremondi_2023} uses guarded type theory to
define a syntactic model for a gradually-typed source
language with dependent types. By working in guarded type theory, they are
able to give an exact semantics to non-terminating computations in a language
which is logically consistent, allowing for metatheoretic results about the
source language to be proven.
Similarly to our approach, they define a guarded lift monad to model potentially-
nonterminating computations and use guarded recursion to model the dynamic type.
However, they do not give a denotational semantics to term precision and it is unclear
how to prove graduality in this setting.
Their work includes a formalization of the syntactic model in Guarded Cubical Agda.

\subsection{Comparison of Denotational and Operational Approaches}

Denotational methods have several benefits compared to operational approaches,
and this carries over to the setting of gradual typing. First, denotational
methods are compositional and reusable, and allow for the use of existing
mathematical constructs and theorems, e.g., partial orderings, monads, etc,
while operational methods tend to require a significant amount of boilerplate
work to be done from scratch in each new development.
As a specific example of the compositional nature of our approach, the treatment
of the cast rules in our work is more compositional than in previous work using
operational semantics. Recall from Section \ref{sec:towards-relational-model}
that the cast rules needed for the proof of graduality build in composition of
type precision derivations. Rather than proving these from scratch in the model,
we are able to take as primitive simpler versions of the cast rules whose
validity in the model is easier to establish. Then from these simpler rules, we
derive the original ones using compositional reasoning.

A further advantage of the denotational approach is that it is completely
independent of any particular syntax of gradual typing. This allows for reuse
across multiple languages and makes it particularly straightforward to
accommodate additions to a language: adding support for a new type amounts to
defining a new object and extending the dynamic type accordingly. It also may be
possible to model alternative cast semantics by changing the definition of the
dynamic type and some type casts.
In contrast, operational methods are not as readily extensible, generally
requiring adding cases to the logical relations and the inductive proofs. In a
mechanized metatheory, this manifests as a ``copy-paste'' reusability rather than
the true code reuse one obtains with a denotational semantics.
An additional benefit to the denotational approach is that it is trivial to
establish the validity of the $\beta$ and $\eta$ principles, because they are
equalities in the semantics, whereas they require tedious manual proof in the
operational approach.

Finally, while we advocate for a denotational approach, the overall structure of
our semantics could be applied to an operational logical relations proof. For
example, there could be a step-indexed logical relation for strong error
ordering and one for weak bisimilarity, corresponding to the fact that objects
in our denotational model have an error ordering and a bisimilarity relation.
Then the precision rules for casts could be similarly validated using
perturbations as we did in this work.

\subsection{Mechanization}\label{sec:mechanization}
In parallel with developing the theory discussed in this paper, we have
developed a partial formalization of our results in Guarded Cubical Agda
\cite{veltri-vezzosi2020}; the formalization is available online \cite{artifact}.
We formalized the major components of the definition of the concrete model
described in Section \ref{sec:concrete-relational-model}: predomains, error
domains, morphisms, relations, squares, using the guarded
features to define the free error domain, the lock-step error ordering, and weak
bisimilarity. We also formalized the no-go theorem from Section
\ref{sec:towards-relational-model}.

Building on these definitions, we formalized the notion of semantic
perturbations and push-pull structures as well as quasi-representable relations,
culminating in the definition of value and computation types and relations as
introduced in Section \ref{sec:concrete-relational-model}.
%
%
We constructed the predomain for the dynamic type via mixed induction and
guarded recursion, and defined its monoid of perturbations. We also defined the
relations $\inat$, $\itimes$, and $\iarr$, and proved that they are
quasi-representable.




Lastly, we have formalized the big-step term semantics discussed in Section
\ref{sec:big-step-term-semantics} and the adequacy of the relational model
discussed in Section \ref{sec:adequacy}. This required us to add axioms about
clock quantification as well as axioms asserting the \emph{clock-irrelevance} of
booleans and natural numbers since as of this writing these axioms are not
built-in to Guarded Cubical Agda. These axioms are discussed in prior work on
guarded type theory \cite{atkey-mcbride2013, kristensen-mogelberg-vezzosi2022}.
The mechanization of adequacy also required us to formalize some essential
lemmas involving clocks and clock-irrelevance; we are considering later
refactoring these as part of a ``standard library'' for Guarded Cubical
Agda.

As of this writing, the remaining formalization work is the following:
\begin{enumerate}
    \item Showing that the functors $\times$ and $\arr$ preserve
    quasi-representability, which requires tedious reasoning about the Kleisli actions of
    $\times$ and $\arr$. These results are proved in the Appendix.
    \item Verifying the rules in the model corresponding to the equations for
    type precision derivations (see the bottom of Figure \ref{fig:gtlc-syntax}).
    These are also included in the Appendix.
    \item Formalizing the syntax-to-semantics translation.
\end{enumerate}

\subsection{Future Work}


Our immediate next step is to apply our approach to give a denotational
semantics to gradually-typed languages with advanced type systems including
higher-order state or runtime-extensible dynamic typing
\cite{DBLP:journals/corr/abs-2210-02169}, as well as richer type disciplines
such as gradual dependent types and effect systems. This will involve adapting
prior operational models based on step-indexing (e.g.,
\cite{new-giovannini-licata-2022}) to the denotational setting and using guarded
type theory to obtain solutions to guarded domain equations as we did in this
work for the denotation of the dynamic type. The generality of our techniques
should allow for many of the constructions to be reused across different
languages. For instance, the free error domain construction can be easily
extended to model effects besides error and stepping
\cite{guarded-interaction-trees, recursion-probabilistic-choice-guarded}.
We also aim to complete our Agda formalization and evolve it into a
reusable framework for mechanized denotational semantics of gradually-typed
languages.

This work has focused on the ``Natural'' semantics of casts , which
validate the full call-by-value $\eta$ laws for types, but other cast
semantics have been proposed such as eager
(\cite{herman-tomb-flanagan-2010}) and transient (\cite{transient})
cast semantics which trade off certain $\eta$ equalities for other
benefits such as early error detection or reduced runtime
overhead \cite{deepandshallowtypes,new-licata-ahmed2019}. It would be a
good test of the generality of our semantic framework to see if it can
model these alternative semantics by adjusting the semantics of types
and casts, and providing proofs of weakened $\eta$ principles.

Additionally, while our focus in this work has been on reasoning up to
weak bisimilarity, the presence of explicit step counting in the model
could be viewed as a form of \emph{cost semantics}, where a runtime
cost is incurred from inspecting the dynamic type. This could possibly
be used to verify that cast optimizations such as space efficient
implementations \cite{herman-tomb-flanagan-2010} are not only
extensionally correct, but have a lower abstract cost.

\bibliographystyle{ACM-Reference-Format}
\begin{DIFnomarkup}
\bibliography{references}
\end{DIFnomarkup}

\newpage
\appendix
\section{Gradual Typing Syntax}\label{sec:appendix-gtlc-syntax}

\newcommand{\dynof}[1]{\textrm{dyn}(#1)}
\begin{theorem}
  For every $A$, there is a derivation $\dynof A : A \ltdyn D$
\end{theorem}
\begin{proof}
  By induction on $A$:
  \begin{enumerate}
  \item $A = \dyn$, then $r(\dyn) : \dyn \ltdyn \dyn$
  \item $A = \nat$, then $\inat : \nat \ltdyn \dyn$
  \item $A = A_1 \ra A_2$ then $(\dynof {A_1} \ra \dynof {A_2})\iarr : A_1 \ra A_2 \ltdyn \dyn$ 
  \item $A = A_1 \times A_2$ then $(\dynof {A_1} \times \dynof{A_2})\itimes : A_1 \times A_2 \ltdyn \dyn$ 
  \end{enumerate}
\end{proof}

\begin{theorem}
  \label{thm:thin}
  For any two derivations $c,c' : A \ltdyn A'$ of the same precision
  $c \equiv c'$
\end{theorem}
\begin{proof}
  \begin{enumerate}
  \item We show this by showing that derivations have a canonical
    form.

    The following presentation of precision derivations has unique derivations
    \begin{mathpar}
      \inferrule{}{\textrm{refl}(D) : D \ltdyn D}\and
      \inferrule{}{\textsf{Inj}_{\text{nat}} : \nat \ltdyn D}\and
      \inferrule{}{\textrm{refl}(\nat) : \nat \ltdyn \nat}\and
      \inferrule{c : A_i \ra A_o \ltdyn D\ra D}{c(\textsf{Inj}_{\text{arr}}) : A_i \ra A_o \ltdyn \nat}\and
      \inferrule{c : A_i \ltdyn A_i' \and d : A_o \ltdyn A_o'}{c \ra d : A_i \ra A_o \ltdyn A_i'\ra A_o'}
    \end{mathpar}
    Since it satisfies reflexivity, cut-elimination and congruence, it
    is a model of the original theory. Since it is a sub-theory of the
    original theory, it is equivalent.
  \end{enumerate}
\end{proof}

As mentioned in Section~\ref{sec:GTLC}, the cast calculus
separates gradual type casts into downcasts and upcasts, whereas many
gradual calculi are built out of a single notion of cast
\[ \inferrule{\Gamma \vdash M : A}{\Gamma \vdash (M :: A') : A' }\]
The term precision rules for this style of calculus are as follows\cite{siek_et_al:LIPIcs:2015:5031}:
\begin{mathpar}
  \inferrule*[Right=CastRight]
  {\Delta \vdash M \ltdyn N_1 : c_1 \and
    c_{1} : A \ltdyn A_1\and
    c_2 : A \ltdyn A_2
  }
  {\Delta \vdash M \ltdyn (N :: A_2) : c_2}

  \inferrule*[Right=CastLeft]
  {\Delta \vdash M_1 \ltdyn N : c_1 \and
    c_1 : A_1 \ltdyn A\and
    c_2 : A_2 \ltdyn A
  }
  {\Delta \vdash (M_1 :: A_2) \ltdyn N : c_2}
\end{mathpar}
New-Ahmed and New-Licata showed that such if the cast $M_1 :: A_2$ for $M_1 : A_1$ is interpreted as
\[ \dnc {(\dynof {A_2})}{\upc {(\dynof{A_1})} M} \]
then the CastRight and CastLeft rules are derivable from the UpL/UpR/DnL/DnR rules
and the \emph{retraction} property of casts: that a downcast after an upcast
is equal to the identity.
In the intensional setting, this retraction property needs to be
weakened: a downcast after an upcast is merely \emph{weakly bisimilar}
to the identity.
But this weak bisimilarity is no longer strong enough to show the validity
of CastRight and CastLeft, due to the lack of transitivity of weak
bisimilarity.

Fortunately there is an alternative argument that makes CastRight and
CastLeft valid if we change slightly the interpretation of casts. To
show this alternative translation we need to first introduce the
\emph{least upper bound} of types:
\begin{lemma}
  The preordered set of types with type precision as ordering has all
  binary least upper bounds defined as follows:
  \begin{align*}
    (A_1 \times A_2) \sqcup (A_1' \times A_2') &= (A_1 \sqcup A_1') \times (A_2 \sqcup A_2')\\
    (A_1 \ra A_2) \sqcup (A_1' \ra A_2') &= (A_1 \ra A_1') \times (A_2 \ra A_2')\\
    \nat \sqcup \nat &= \nat\\
    A \sqcup A' &= \dyn \text{ otherwise }
  \end{align*}
\end{lemma}

Then for $M_1 : A_1$ we translate $M_1 :: A_2$ to
\[ \dnc {c^\sqcup_2}\upc{c^\sqcup_1} M \]
where $c^\sqcup_1 : A_1 \ltdyn A_1\sqcup A_2$ and $c^\sqcup_2 : A_2 \ltdyn
A_1\sqcup A_2$. In the presence of the strong retraction property, New
and Ahmed showed that this is equal to the prior interpretation. With
a weak retraction property, it is still weakly bisimilar to the prior
interpretation, so for the extensional properties we care about
ultimately this semantics is equivalent. However our compositional
arguments do care about this difference, and so we use instead this
translation that avoids the need for any retraction:
\begin{theorem}
  If $M_1 :: A_2$ is interpreted as $\dnc {c^\sqcup_2}\upc{c^\sqcup_1} M$,
  then CastRight and CastLeft are admissible.
\end{theorem}
\begin{proof}
  We show CastLeft, as this is the one that required retraction in
  prior work. CastRight follows similarly.

  Assume $M_1 \ltdyn N : c_1$, we seek to prove that
  \[ \dnc {c^\sqcup_2}\upc{c^\sqcup_1} M \ltdyn N : c_2 \]
  Since $A \sqcup A_2$ is the
  least upper bound, we have there must be a derivation $c^{\sqcup} :
  A_1\sqcup A_2 \ltdyn A$.
  Further by Theorem~\ref{thm:thin}, we have
  $c^\sqcup_2c^{\sqcup} \equiv c_2$, so it suffices to show
  \[M_1 \ltdyn \dnc {c^\sqcup_2}\upc{c^\sqcup} N : c^\sqcup_2c^{\sqcup} \]
  Then applying DnL and UpL it is sufficient to show
  \[ M_1 \ltdyn N : c^\sqcup_1c^\sqcup \]
  But again by Theorem~\ref{thm:thin} we have $c^\sqcup_1c^\sqcup
  \equiv c_1$ so this follows by our assumption that $M_1 \ltdyn N
  :c_1$.
\end{proof}

\section{Kleisli Actions}\label{sec:kleisli-actions}

We now define the Kleisli actions of the type constructors. These are the following
\begin{enumerate}
\item $- \tok B : \op\PreDom_k \to \ErrDom_k$
\item $A \tok - : \ErrDom_k \to \ErrDom_k$
\item $- \timesk A_2 : \PreDom_k \to \PreDom_k$
\item $A_1 \timesk - : \PreDom_k \to \PreDom_k$
\end{enumerate}
While these can be defined in any CBPV model, for simplicity we just
give their definition explicitly for predomains and error domains:
\begin{definition}{~}
  Given $\phi : \li A_i \multimap \li A_o$, we define $\phi \tok B : U(A_o \to B) \to U(A_i \to B)$ as
  \[ (\phi \tok B)(f)(x) = f^\dagger(\phi(\eta(x))) \]

  This is functorial in that $\id \tok B = \id$ and $(\phi \circ
  \phi') \tok B = (\phi' \tok B) \circ (\phi \tok B)$.

  Further, this preserves squares in that if $\phi \ltsq{\li c_i}{\li c_o} \phi'$ then $\phi \tok B \ltsq{U(c_o \to d)}{U(c_i \to d)} \phi' \tok B'$.
\end{definition}
\begin{proof}
  {~}
  \begin{itemize}
  \item For identity
    \[ (\id \tok B)(f)(x) = f^\dagger(\eta(x)) = f(x) \]

  \item 
    For composition, expanding definitions it suffices to show
    \[ f^\dagger(\phi(\phi'(\eta(x)))) = ((\phi \tok B)(f))^\dagger(\phi'(\eta(x))) \]
    Which follows if we can show
    \[ ((\phi \tok B)(f))^\dagger = f^\dagger \circ \phi \]
    By the freeness of $\li -$ it suffices to show it for inputs of the form $\eta y$:
    \[ ((\phi \tok B)(f))^\dagger(\eta y) = ((\phi \tok B)(f))(y) = f^\dagger(\phi(\eta y))\]
  \item For squares, we assume $f \binrel{U(c_o\to d)} f'$ and $x
    \binrel {c_i} x'$ and we need to show $f^\dagger(\phi(x)) \binrel
    d (f')^\dagger(\phi'(x))$. This follows easily since $f^\dagger
    \ltsq{\li c_o}{d} f^{\prime\dagger}$.
  \end{itemize}
\end{proof}

\begin{definition}
  Given $f : UB \to UB'$, define $A \tok f : U(A \to B) \to U(A \to B')$ as
    \[ (A \tok f)(g)(x) = f(g(x)) \]

  This is functorial in that $A \tok \id = \id$ and $A \tok (f \circ
  f') = (A \tok f) \circ (A \tok f')$.

  Further this preserves squares in that if $f \ltsq{Ud_i}{Ud_o} f'$
  then $A \tok f \ltsq{U(c \to d_i)}{U(c \to d_o)} A' \tok f'$
\end{definition}
\begin{proof}
  \begin{itemize}
  \item For identity
    \[ (A \tok \id)(g)(x) = g(x) \]
  \item 
    for composition
    \[  (A \tok (f \circ f'))(g)(x) = f(f'(g(x))) = (A\tok f)((A \tok f')(g))(x) \]
  \item 
   for squares, given $g \binrel{U(c \to d_i)} g'$ and $x \binrel{c}
   x'$ we have $f(g(x)) \binrel{d_o} f'(g'(x'))$.
  \end{itemize}
\end{proof}

\begin{definition}[Kleisli Actions of $\times$]
  Given $\phi_1 : \li A_1 \multimap \li A_1'$ we define $\phi_1 \timesk A_2 : \li (A_1 \times A_2) \multimap (A_1' \times A_2)$ as the unique extension of the following to a homomorphism:
  \[ (\phi_1 \timesk A_2)(\eta (x_1,x_2)) = (\lambda x_1'. \eta(x_1',x_2))^\dagger(\phi_1(\eta(x_1))) \]
  and similarly define
  \[ (A_1 \timesk \phi_2)(\eta (x_1,x_2)) = (\lambda x_2'. \eta(x_1,x_2'))^\dagger(\phi_2(\eta(x_2)))\]
\end{definition}
\begin{proof}
  We show cases only for $-\timesk A_2$ as the other is symmetric.
  \begin{itemize}
  \item For identity it is sufficient to consider inputs of the form $\eta(x_1,x_2)$, then indeed
    \[ (\id \timesk A_2)(\eta(x_1,x_2)) = (\lambda x_1'. \eta(x_1',x_2))^\dagger(\eta(x_1)) = \eta(x_1,x_2) \]
  \item For composition, again for pure inputs this reduces to
    \[ (\lambda x_1'. \eta(x_1',x_2))^\dagger(\phi_1'(\phi_1(\eta(x_1))))
    = (\phi_1'\timesk A_2)((\lambda x_1'. \eta(x_1',x_2))^\dagger(\phi_1(\eta(x_1)))) \]
    So it suffices to show
    \[ (\lambda x_1'. \eta(x_1',x_2))^\dagger \circ \phi_1' = (\phi_1'\timesk A_2) \circ (\lambda x_1'. \eta(x_1',x_2))^\dagger \]
    Since both sides are homomorphisms out of the free error domain by freeness it suffices to show they are equal for inputs of the form $\eta(x_1')$:
    \[  (\lambda x_1'. \eta(x_1',x_2))^\dagger(\phi_1'(\eta(x_1')))
    = (\phi_1'\timesk A_2)(\eta(x_1',x_2))
    = (\phi_1'\timesk A_2)((\lambda x_1'. \eta(x_1',x_2))^\dagger(\eta(x_1')))
    \]

  \item Finally, for squares, we need to show if $\phi_l \ltsq{\li
    c_1}{\li c_1'} \phi_r$ then $\phi_l \timesk A_{2l} \ltsq{\li (c_1
    \times c_2)}{\li (c_1' \times c_2)} \phi_r \timesk A_{2r}$.

    By the universal property of $\li-$ on relations, it suffices to
    assume $x_{1l} \binrel{c_1} x_{1r}$ and $x_{2l} \binrel{c_2} x_{2r}$ and prove
    \[ (\phi_l \timesk A_{2l})(\eta(x_{1l},x_{2l})) \binrel{(\li (c_1' \times c_2))} (\phi_r \timesk A_{2r})(\eta(x_{1r},x_{2r})) \]
    expanding definitions this is
    \[ (\lambda x_{1l}'. \eta(x_{1l}',x_{2l}))^\dagger\phi_l(\eta(x_{1l}))
    \binrel{(\li (c_1' \times c_2))}
    (\lambda x_{1r}'. \eta(x_{1r}',x_{2r}))^\dagger\phi_r(\eta(x_{1r}))
    \]
    By our assumptions we have $\phi_l(\eta(x_{1l}) \binrel {\li c_1'} \phi_r(\eta(x_{1r})$ so since $-^\dagger$ preserves squares it suffices to show
    \[ (\lambda x_{1l}'. \eta(x_{1l}',x_{2l})) \ltsq{\li c_1}{\li (c_1' \times c_2)} (\lambda x_{1r}'. \eta(x_{1r}',x_{2r})) \]
    To show this assume $x_{1l}' \binrel{c_1'} x_{1r}'$. Then we have
    \[ \eta(x_{1l}',x_{2l}) \binrel{\li (c_1' \times c_2)} \eta(x_{1r'},x_{2l})  \]
    holds.
  \end{itemize}
\end{proof}

\section{Omitted Proofs Section \ref{sec:towards-relational-model}}

\begin{theorem}
  Let $R$ be a binary relation on the free error domain $U(\li A)$. If
  $R$ satisfies the following properties
  \begin{enumerate}
  \item Transitivity
  \item $\theta$-congruence: If $\later_t (\tilde{x}_t \binrel{R} \tilde{y}_t)$, then $\theta(\tilde{x}) \binrel{R} \theta(\tilde{y})$.
  \item Right step-insensitivity: If $x \binrel{R} y$ then $x \binrel{R} \delta y$.
  \end{enumerate}
  Then for any $l : U(\li A)$, we have $l \binrel{R} \Omega$. If $R$ is left
  step-insensitive instead then $\Omega \binrel{R} x$.
\end{theorem}
\begin{proof}
  By \lob-induction: we assume that $\laterhs (l \binrel{R} \Omega)$, and we
  show $(l \binrel{R} \Omega)$. We have $l \binrel{R} \delta l$ by right
  step-insensitivity.
  By transitivity, it will therefore suffice to show that $\delta l \binrel{R}
  \Omega$.
  Then since $\Omega = \delta \Omega$, it suffices to show $\delta l \binrel{R}
  \delta\Omega$.
  By $\theta$-congruence, it suffices to show that
  $\later_t [(\nxt\, l)_t \binrel{R} (\nxt\, \Omega)_t]$, i.e., $\later (l \binrel{R} \Omega)$,
  which is our induction hypothesis.

  The proof when $R$ is left step-insensitive is symmetric.
\end{proof}

\section{Details of the Relational Model Construction}

In Section \ref{sec:concrete-relational-model} we described our model of gradual
typing. We omitted several technical proofs involving the well-definedness of
composition of relations and the actions of functors. We provide those proofs in
this section.







\subsection{The Dynamic Type}
The ordering $\le_D$ on the predomain $D$ is defined by
\begin{align*}
  \inat(n) \le \inat(n') 
      &\iff n = n' \\
  \itimes (d_1, d_2) \le \itimes (d_1', d_2')
      &\iff d_1 \le_D d_2 \text{ and } d_1' \le_D d_2'\\
  \iarr(\tilde{f}) \le \iarr(\tilde{f'}) 
      &\iff \later_t(\tilde{f}_t \le \tilde{f'}_t),
\end{align*}

and the bisimilarity relation is analogous.







\subsection{Lemmas about Perturbations}

The goal of this section is to prove the following lemmas:

\begin{lemma}\label{lem:push-pull-comp}
  Let $(A_1, M_{A_1}, i_{A_1})$, $(A_2, M_{A_2}, i_{A_2})$, and $(A_3, M_{A_3},
  i_{A_3})$ be value objects, and let $c : A_1 \rel A_2$ and $c' : A_2 \rel A_3$
  be predomain relations.

  Given a push-pull structure $\Pi_c$ for $c$ and $\Pi_{c'}$ for $c'$, we can
  define a push-pull structure $\Pi_{c \comp c'}$ for $c \comp c'$.

  Likewise, we can define a push-pull structure for the composition of error
  domain relations.
\end{lemma}
\begin{proof}
  We define $\piv_{c \comp c'}$ as follows: We first define
  $\push_{c \comp c'} = \push_{c'} \circ \push_{c}$ and 
  $\pull_{c \comp c'} = \pull_{c} \circ \pull_{c'}$.
  We then observe that the required squares exist for both push and pull. In
  particular, for push we have that $i_{A_1}(m_{A_1}) \ltdyn_c^c
  i_{A_2}(\push_c(m_{A_1}))$ using the push property for $c$, and then using the
  push property for $c'$ we have $i_{A_2}(\push_c(m_{A_1})) \ltdyn_{c'}^{c'}
  i_{A_3}(\push_{c'}(\push_c(m_{A_1})))$. We can then compose these squares
  horizontally to obtain the desired square. The pull property follows
  similarly.

  The push-pull structure for the composition of computation relations is
  defined analogously.
\end{proof}

\begin{lemma}\label{lem:push-pull-U-F} 
  
  Let $A$ and $A'$ be value objects and let $c : A \rel A'$ be a relation on the
  underlying predomains. Given a push-pull structure $\Pi_c$ for $c$ we can
  define a push-pull structure $\Pi_{\li c}$ for $\li c$ with respect to $\li A$
  and $\li A'$.

  Likewise, let $B$ and $B'$ be computation objects and let $d : B \rel B'$ be
  an error domain relation. Given a push-pull structure for $d$ with respect to
  $B$ and $B'$, we can define a push-pull structure for $Ud$ with respect to
  $UB$ and $UB'$.
\end{lemma}
\begin{proof}
  
  We define $\push_{\li c} : \mathbb{N} \oplus M_{A} \to \mathbb{N} \oplus
  M_{A'}$ by the universal property of the coproduct of monoids. In particular,
  it suffices to define a homomorphism of monoids $\mathbb{N} \to \mathbb{N}
  \oplus M_{A'}$ and $M_{A} \to \mathbb{N} \oplus M_{A'}$.

  The former is simply $\inl$ and the latter is $\inr \circ \push_c$.
  
  Then to establish the push property, it suffices by the universal property of
  the coproduct, and the fact that $\mathbb{N}$ is the free monoid on one
  generator, to show that the following two squares exist:

\[\begin{tikzcd}[ampersand replacement=\&]
	{\li A} \& {\li A'} \\
	{\li A} \& {\li A'}
	\arrow["{\li c}", "\shortmid"{marking}, no head, from=1-1, to=1-2]
	\arrow["{\delta^*_A}"', from=1-1, to=2-1]
	\arrow["{\delta^*_{A'}}", from=1-2, to=2-2]
	\arrow["{\li c}"', "\shortmid"{marking}, no head, from=2-1, to=2-2]
\end{tikzcd}\]

\[\begin{tikzcd}[ampersand replacement=\&]
	{\li A} \& {\li A'} \\
	{\li A} \& {\li A'}
	\arrow["{\li c}", "\shortmid"{marking}, no head, from=1-1, to=1-2]
	\arrow["{\li(i_A(m_A))}"', from=1-1, to=2-1]
	\arrow["{\li (i_{A'}(\push_c(m_A)))}", from=1-2, to=2-2]
	\arrow["{\li c}"', "\shortmid"{marking}, no head, from=2-1, to=2-2]
\end{tikzcd}\]

where $\delta* = (\delta \circ \eta)^\dagger$

The first square exists by the fact that the monadic extension operation
$-^\dagger$ is monotone in its argument, and the fact that there is a
square $\delta_A \circ \eta_A \ltsq{}{} \delta_{A'} \circ \eta_{A'}$.

The second square exists by the push property for $c$ and the fact that $\li$
acts on squares.

The proof for the pull property is dual, and the construction of a push-pull
structure for $Ud$ is analogous.
\end{proof}

\begin{lemma}\label{lem:push-pull-times} Let $A_1$, $A_1'$, $A_2$, and $A_2'$ be
  value objects. Let $c_1 : A_1 \rel A_1'$ and $c_2 : A_2 \rel A_2'$ be
  relations on the underlying predomains. Given push-pull structures $\Pi_{c_1}$
  for $c_1$ and $\Pi_{c_2}$ for $c_2$ we can define a push-pull structure
  $\Pi_{c_1 \times c_2}$ for $c_1 \times c_2$.
\end{lemma}
\begin{proof}
Recall that the monoid of syntactic perturbations for $A_1 \times A_2$ is
$M_{A_1} \oplus M_{A_2}$ and the monoid for $A_1' \times A_2'$ is $M_{A_1'}
\oplus M_{A_2'}$. Thus to define the push homomorphism for $c_1 \times c_2$, it
suffices by the universal property of the coproduct of monoids to define a
homomorphism $M_{A_1} \to M_{A_1'} \oplus M_{A_2'}$ and $M_{A_2} \to M_{A_1'}
\oplus M_{A_2'}$. The former is $\inl \circ \push_{c_1}$, and the latter is
$\inr \circ \push_{c_2}$. The pull homomorphism is defined similarly.

For the push property, we need to show that there is a square 
$i_{A_1 \times A_2}(m_1) \ltsq{}{} i_{A_1' \times A_2'}(\push_{c_1 \times c_2}(m_1))$.
By the universal property of the coproduct, it suffices to consider two cases: a
perturbation $m_1 \in M_{A_1}$ and a perturbation $m_2 \in M_{A_2}$. In the
first case, using the definition of the interpretation homomorphism for the
product, the square becomes $(i_{A_1}(m_1) \times \id_{A_2}) \ltsq{}{}
(i_{A_1'}(\push_{c_1}(m_1)) \times \id_{A_2'})$.
Then by the action of $\times$ on squares, it suffices to show that there are
squares $i_{A_1}(m_1) \ltsq{}{} i_{A_1'}(\push_{c_1}(m_1))$ and $\id_{A_2}
\ltsq{}{} \id_{A_2'}$. The former is the push property for $c_1$, and the latter
is an identity square.

The second case, i.e., $m_2 \in M_{A_2}$, is symmetric, and the proof of the
pull property is dual.
\end{proof}

\begin{lemma}\label{lem:push-pull-arrow} Let $A$ and $A'$ be value objects and
  $B$ and $B'$ be computation objects. Let $c : A \rel A'$ be a relation on the
  underlying predomains and $d : B \rel B'$ a relation on the underlying error
  domains. Given push-pull structures $\Pi_c$ for $c$ and $\Pi_d$ for $d$, we
  can define a push-pull structure $\Pi_{c \arr d}$ for $c \arr d$.
\end{lemma}
\begin{proof}
  Similar to the product case above.
\end{proof}

We next define the actions of the Kleisli arrow and product functors on
syntactic perturbations.


\begin{definition}\label{def:kleisli-arrow-perturbations} Let $A$ be a value
  object and $B$ be a computation object. We define a homomorphism of monoids
  $id \tok - \colon M_{UB} \to M_{U(A \arr B)}$ via the universal property of the
  coproduct of monoids by sending $\mathbb{N}$ to the first injection in $M_{U(A
  \arr B)} = \mathbb{N} \oplus M_A^{op} \oplus M_B$, and sending $M_B$ to the
  third injection.

  Likewise, we define $- \tok \id \colon M_{\li A}^{op} \to M_{U(A \arr B)}$ as
  follows. First note that $M_{\li A}^{op} = \mathbb{N} \oplus M_A^{op}$. Then
  we construct the homomorphism via the universal property, sending $\mathbb{N}$
  to the first injection $M_A^{op}$ to the second injection.
\end{definition}

We can similarly construct a Kleisli product operation on perturbations:

\begin{definition}\label{def:kleisli-product-perturbations}
  Let $A_1$ and $A_2$ be value objects. We define a homomorphism of monoids 
  $- \timesk \id \colon M_{\li A_1} \to M_{\li (A_1 \times A_2)}$ using the universal property,
  sending $\mathbb{N}$ to the first injection and $M_{A_1}$ to the second injection.

  Likewise, we define $\id \timesk - \colon M_{\li A_2} \to M_{\li (A_1 \times A_2)}$
  by sending $\mathbb{N}$ to the first injection and $M_{A_2}$ to the third injection.

\end{definition}

It is then easy to verify that the Kleisli arrow action on perturbation is well-defined 
in that for any $m \in M_{UB}$, we have
\[ \ptb_{U(A \arr B)}(\id \tok m) = \id \tok i_{UB}(m), \]
where the LHS is the Kleisli action on perturbation and the RHS is the Kleisli
action on morphisms. 
The proof uses the fact that the morphism $\delta$ commutes with all error
domain morphisms, which is a consequence of the definition of error domain
morphism.
The other verification involving the Kleisli arrow action is similar, as are the
two Kleisli product actions.

\subsection{Lemmas involving Quasi-Representable Relations}

In this section we prove lemmas about quasi-representability needed for showing
that our notions of value and computation relation are compositional, and for
defining the action of the functors $U$, $\li$, $\arr$, and $\times$ on value
and computation relations.

Recall the notion of quasi-equivalence of relations as defined in Definition
\ref{def:quasi-equivalent}. The following lemma will be useful in showing
that two relations are quasi-equivalent.

\begin{lemma}\label{lem:left-rep-by-same-morphism} Let $A$ and $A'$ be value
  objects, and let $c, c' : A \rel A'$ be relations between the underlying
  predomains. If $c$ and $c'$ are both quasi-left-representable by the same
  predomain morphism $f : A \to A'$, then $c \qordeq c'$. If $c$ and $c'$ are
  both quasi-right-representable by the same predomain morphism $g : A' \to A$
  then $c \qordeq c'$.

  Dually, let $B$ and $B'$ be computation objects, and let $d, d' : B \rel B'$
  be relations between the underlying error domains. If $d$ and $d'$ are both
  quasi-right-representable by the same error domain morphism $\phi : B'
  \multimap B$, then $d \qordeq d'$. If $d$ and $d'$ are both
  quasi-left-representable by the same $\phi : B \multimap B'$
\end{lemma}
\begin{proof}
  We show the result for $c$ and $c'$ that are quasi-left-representable by the same $f$;
  the other proofs are analogous. 

  By $\upr$ for $c'$, there exists a perturbation $\delle_{c'}$ and a square
  $\upr_{c'} : \delle_{c'} \ltdyn_{c'}^{r(A)} f$.

  By $\upl$ for $c$, there exists a perturbation $\delre_c$ and a square 
  $\upl_c : f \ltdyn_{r(A')}^{c} \delre_c$.
  
  Composing these horizontally we get the following square:

  \[\begin{tikzcd}[ampersand replacement=\&]
    A \& A \& {A'} \\
    A \& {A'} \& {A'}
    \arrow["{r(A)}", "\shortmid"{marking}, no head, from=1-1, to=1-2]
    \arrow["{i_A(\delle_{c'})}"', from=1-1, to=2-1]
    \arrow["c", "\shortmid"{marking}, no head, from=1-2, to=1-3]
    \arrow["f"', from=1-2, to=2-2]
    \arrow["{i_{A'}(\delre_c)}", from=1-3, to=2-3]
    \arrow["{c'}"', "\shortmid"{marking}, no head, from=2-1, to=2-2]
    \arrow["{r(A')}"', "\shortmid"{marking}, no head, from=2-2, to=2-3]
  \end{tikzcd}\]

  Then since $r(A) \comp c = c$ and $c' \comp r(A') = c'$, we obtain the desired square.
  The other square (i.e., with $c'$ on top) is constructed in an analogous
  manner.

\end{proof}

Next we show that the functors $\li$ and $U$ preserve quasi-representability.

\begin{lemma}\label{lem:representation-U-F}
  Let $A$ and $A'$ be value objects and let $c : A \rel A'$ be a predomain relation.
  If $c$ is quasi-left-representable, then $\li c$ is quasi-left-representable.
  Likewise, if $c$ is quasi-right-representable, then $\li c$ is quasi-right-representable.


  Similarly, let $B$ and $B'$ be computation objects and let $d : B \rel B'$ be an error domain relation.
  If $d$ is quasi-right-representable, then $Ud$ is quasi-right-representable.

\end{lemma}
\begin{proof}
  To show that $\li c$ is quasi-left-representable, we define
  \begin{itemize}
    \item $e_{\li c} = \li e_c$
    \item $\delre_{\li c} = \inr (\delre_c)$ 
    (recalling that the monoid $M_{\li A'} = \mathbb{N} \oplus M_{A'}$)
    \item $\delle_{\li c} = \inr (\delle_c)$
  \end{itemize}
  We obtain the two squares $\upl$ and $\upr$ using the definition of the
  interpretation of syntactic perturbations for $\li$ and the functorial action
  of $\li$ on the corresponding squares for $c$.

  The proof for $Ud$ is analogous.
\end{proof}


Next we show that the reflexive relation is quasi-representable.

\begin{lemma}\label{lem:reflexive-rel-quasi-rep}
  Let $A$ be a value object. Then $r(A)$ is quasi-left-representable and quasi-right-representable.
  Likewise, $r(B)$ is quasi-left- and quasi-right-representable for any computation object $B$.
\end{lemma}
\begin{proof}
  To show $r(A)$ is quasi-left-representable, we take the embedding to be the
  identity morphism and the perturbations $\delle = \delre = \id_{M_{A}}$, i.e.,
  the identtiy of the monoid. Then because $i_A$ is a homomorphism of monoids,
  it sends the identity element to the identity morphism, so the $\upl$ and
  $\upr$ squares are just identity squares.

  The same argument shows that $r(A)$ is quasi-right-representable and that
  $r(B)$ is quasi-left- and quasi-right-representable.
\end{proof}


We begin by showing that the composition of quasi-representable relations with
push-pull structures is quasi-representable.

\begin{lemma}\label{lem:representation-comp}
  Let $A$, $A'$, and $A''$ be value objects, and $B$, $B'$, and
  $B''$ be computation objects. Let $c : A \rel A'$ and $c' : A' \rel A''$ be
  predomain relations with push-pull structures, and let $d : B \rel B'$ and $d'
  : B' \rel B''$ be error domain relations with push-pull structures.
  \begin{enumerate}
    \item If $c$ and $c'$ are quasi-left- (resp. right)-representable, then $c
    \comp c'$ is quasi-left- (resp. right)-representable.

    \item If $d$ and $d'$ are quasi-left- (resp. right)-representable, then
    $d \comp d'$ is quasi-left (resp. right)-representable.
  \end{enumerate}
\end{lemma}
\begin{proof}
    \item To show $c \comp c'$ is quasi-left-representable we define
    \begin{itemize}
      \item $e_{c \comp c'} = e_{c'} \circ e_c$
      \item $\delre_{c \comp c'} = \delre_{c'} \cdot \push_{c'}(\delre_c)$ where
      $\cdot$ denotes multiplication in the monoid $M_A$
      \item $\delle_{c \comp c'} = \pull_c(\delle_{c'}) \cdot \delle_c$
      \item $\upl$ is the following square: 
\[\begin{tikzcd}[ampersand replacement=\&,column sep=3.15em]
	A \& {A'} \& {A''} \\
	{A'} \& {A'} \& {A''} \\
	{A'} \&\& {A''} \\
	{A''} \&\& {A''}
	\arrow[""{name=0, anchor=center, inner sep=0}, "{e_c}"', from=1-1, to=2-1]
	\arrow[""{name=1, anchor=center, inner sep=0}, "\id"', from=2-1, to=3-1]
	\arrow[""{name=2, anchor=center, inner sep=0}, "{e_{c'}}"', from=3-1, to=4-1]
	\arrow[""{name=3, anchor=center, inner sep=0}, "\id", from=2-3, to=3-3]
	\arrow["{r(A')}", "\shortmid"{marking}, no head, from=2-1, to=2-2]
	\arrow["{c'}", "\shortmid"{marking}, no head, from=2-2, to=2-3]
	\arrow[""{name=4, anchor=center, inner sep=0}, "{\push_{c'}(\delre_c)}", from=1-3, to=2-3]
	\arrow["c", "\shortmid"{marking}, no head, from=1-1, to=1-2]
	\arrow["{c'}", "\shortmid"{marking}, no head, from=1-2, to=1-3]
	\arrow[""{name=5, anchor=center, inner sep=0}, "{\delre_c}", from=1-2, to=2-2]
	\arrow[""{name=6, anchor=center, inner sep=0}, "{\delre_{c'} }", from=3-3, to=4-3]
	\arrow["{c'}", "\shortmid"{marking}, no head, from=3-1, to=3-3]
	\arrow["{r(A'')}"', "\shortmid"{marking}, no head, from=4-1, to=4-3]
	\arrow["{\upl_c}"{description}, draw=none, from=0, to=5]
	\arrow["{(**)}"{description}, draw=none, from=1, to=3]
	\arrow["{(*)}"{description}, draw=none, from=5, to=4]
	\arrow["{\upl_{c'}}"{description}, draw=none, from=2, to=6]
\end{tikzcd}\]

      The square $(*)$ exists by the push-pull property for $c'$, and the square
      $(**)$ exists by the downward closure of $c'$.

      \item $\upr$ is the following square: 
\[\begin{tikzcd}[ampersand replacement=\&,column sep=3.15em]
	A \&\& A \\
	A \&\& {A'} \\
	A \& {A'} \& {A'} \\
	A \& {A'} \& {A''}
	\arrow[""{name=0, anchor=center, inner sep=0}, "{\delle_c}"', from=1-1, to=2-1]
	\arrow[""{name=1, anchor=center, inner sep=0}, "\id"', from=2-1, to=3-1]
	\arrow[""{name=2, anchor=center, inner sep=0}, "{\pull_c(\delle_{c'})}"', from=3-1, to=4-1]
	\arrow[""{name=3, anchor=center, inner sep=0}, "\id", from=2-3, to=3-3]
	\arrow[""{name=4, anchor=center, inner sep=0}, "{e_c}", from=1-3, to=2-3]
	\arrow[""{name=5, anchor=center, inner sep=0}, "{e_{c'}}", from=3-3, to=4-3]
	\arrow["{r(A)}", "\shortmid"{marking}, no head, from=1-1, to=1-3]
	\arrow["c", "\shortmid"{marking}, no head, from=2-1, to=2-3]
	\arrow["c", "\shortmid"{marking}, no head, from=3-1, to=3-2]
	\arrow["{r(A')}", "\shortmid"{marking}, no head, from=3-2, to=3-3]
	\arrow["c"', "\shortmid"{marking}, no head, from=4-1, to=4-2]
	\arrow["{c'}"', "\shortmid"{marking}, no head, from=4-2, to=4-3]
	\arrow[""{name=6, anchor=center, inner sep=0}, "{\delle_{c'}}"', from=3-2, to=4-2]
	\arrow["{(*)}"{description}, draw=none, from=1, to=3]
	\arrow["{\upr_c}"{description}, draw=none, from=0, to=4]
	\arrow[draw=none, from=6, to=5]
	\arrow["{(**)}"{description}, draw=none, from=2, to=6]
	\arrow["{\upr_{c'}}"{description}, draw=none, from=6, to=5]
\end{tikzcd}\]
    \end{itemize}

    To show $c \comp c'$ is quasi-right-representable we define
    \begin{itemize}
      \item $p_{c \comp c'} = p_c \circ p_{c'}$
      \item $\dellp_{c \comp c'} = \dellp_c \cdot \pull_c(\dellp_{c'})$
      \item $\delrp_{c \comp c'} = \push_{c'}(\delrp_c) \cdot \delrp_{c'}$
      \item $\dnr$ is the following square: 
\[\begin{tikzcd}[ampersand replacement=\&]
	A \& {A'} \& {A''} \\
	A \& {A'} \& {A'} \\
	A \&\& {A'} \\
	A \&\& A
	\arrow["c", "\shortmid"{marking}, no head, from=1-1, to=1-2]
	\arrow[""{name=0, anchor=center, inner sep=0}, "{i_A(\pull_c(\dellp_{c'}))}"', from=1-1, to=2-1]
	\arrow["{c'}", "\shortmid"{marking}, no head, from=1-2, to=1-3]
	\arrow[""{name=1, anchor=center, inner sep=0}, "{\dellp_{c'}}"', from=1-2, to=2-2]
	\arrow[""{name=2, anchor=center, inner sep=0}, "{p_{c'}}", from=1-3, to=2-3]
	\arrow["c", "\shortmid"{marking}, no head, from=2-1, to=2-2]
	\arrow["\id"', from=2-1, to=3-1]
	\arrow["{r(A')}", "\shortmid"{marking}, no head, from=2-2, to=2-3]
	\arrow["\id", from=2-3, to=3-3]
	\arrow["c", "\shortmid"{marking}, no head, from=3-1, to=3-3]
	\arrow[""{name=3, anchor=center, inner sep=0}, "{i_A(\dellp_d)}"', from=3-1, to=4-1]
	\arrow[""{name=4, anchor=center, inner sep=0}, "{p_c}", from=3-3, to=4-3]
	\arrow["{r(A)}"', "\shortmid"{marking}, no head, from=4-1, to=4-3]
	\arrow["{(*)}"{description}, draw=none, from=0, to=1]
	\arrow["{\dnr_{c'}}"{description}, draw=none, from=1, to=2]
	\arrow["{\dnr_c}"{description}, draw=none, from=3, to=4]
\end{tikzcd}\]
      
      Here the square $(*)$ exists by the push-pull property for $c$.
      
      \item $\dnl$ is the following square: 
\[\begin{tikzcd}[ampersand replacement=\&]
	{A''} \&\& {A''} \\
	{A'} \&\& {A''} \\
	{A'} \& {A'} \& {A''} \\
	A \& {A'} \& {A''}
	\arrow["{r(A'')}", "\shortmid"{marking}, no head, from=1-1, to=1-3]
	\arrow[""{name=0, anchor=center, inner sep=0}, "{p_{c'}}"', from=1-1, to=2-1]
	\arrow[""{name=1, anchor=center, inner sep=0}, "{i_{A''}(\delrp_{c'})}", from=1-3, to=2-3]
	\arrow["{c'}", "\shortmid"{marking}, no head, from=2-1, to=2-3]
	\arrow["\id"', from=2-1, to=3-1]
	\arrow["\id", from=2-3, to=3-3]
	\arrow["{r(A')}", "\shortmid"{marking}, no head, from=3-1, to=3-2]
	\arrow[""{name=2, anchor=center, inner sep=0}, "{p_c}"', from=3-1, to=4-1]
	\arrow["{c'}", "\shortmid"{marking}, no head, from=3-2, to=3-3]
	\arrow[""{name=3, anchor=center, inner sep=0}, "{\delrp_c}", from=3-2, to=4-2]
	\arrow[""{name=4, anchor=center, inner sep=0}, "{i_{A''}(\push_{c'}(\delrp_c))}", from=3-3, to=4-3]
	\arrow["c"', "\shortmid"{marking}, no head, from=4-1, to=4-2]
	\arrow["{c'}"', "\shortmid"{marking}, no head, from=4-2, to=4-3]
	\arrow["{\dnl_{c'}}"{description}, draw=none, from=0, to=1]
	\arrow["{\dnl_c}"{description}, draw=none, from=2, to=3]
	\arrow["{(*)}"{description}, draw=none, from=3, to=4]
\end{tikzcd}\]
    \end{itemize}

    \item Same as (1), replacing predomains/morphisms/relations/squares with error domains.
\end{proof}

The next two lemmas concern quasi-equivalence and the functors $\li$ and
$U$.

\begin{lemma}\label{lem:Fcc-equiv-FcFc'}
  
  Let $A$, $A'$, and $A''$ be value objects and let $c : A \rel A'$ and $c' : A'
  \rel A''$ be predomain relations. Then we have $\li(c \comp c') \bisim \li c
  \comp \li c'$.
\end{lemma}
\begin{proof}
  First, we claim that $\li(c \comp c')$ and $\li c \comp \li c'$ are both
  quasi-left-represented by $\li e_{c'} \circ \li e_c$. Indeed, we have by part
  (1) of Lemma \ref{lem:representation-comp} that $e_c' \circ e_c$
  quasi-left-represents $c \comp c'$, and then by Lemma
  \ref{lem:representation-U-F} we have $\li(e_{c'} \circ e_c) = \li e_{c'} \circ \li e_c$
  quasi-left-represents $\li (c \comp c')$.
  On the other hand, we also know that $\li e_c$ quasi-left-represents $\li c$ and
  $\li e_{c'}$ quasi-left-represents $\li c'$ again by Lemma
  \ref{lem:representation-U-F}. Then by part (2) of Lemma
  \ref{lem:representation-comp}, their composition quasi-left-represents $\li c
  \comp \li c'$.

  The result now follows by Lemma \ref{lem:left-rep-by-same-morphism}.
\end{proof}

\begin{lemma}\label{lem:Udd-equiv-UdUd'}
  Let $B$, $B'$, and $B''$ be value objects and let $d : B \rel B'$ and $d' : B'
  \rel B''$ be error domain relations. Then we have $U(d \comp d') \bisim Ud \comp Ud'$.
\end{lemma}
\begin{proof}
  Analogous to the proof of the previous lemma.
\end{proof}

Now we combine the above lemmas to show a result about the functors $\li$ and
$U$ applied to a composition of relations:

\begin{lemma}\label{lem:representation-comp-F-U}
  Let $A$, $A'$, and $A''$ be value objects, and $B$, $B'$, and
  $B''$ be computation objects. Let $c : A \rel A'$ and $c' : A' \rel A''$ be
  predomain relations with push-pull structures, and let $d : B \rel B'$ and $d'
  : B' \rel B''$ be error domain relations with push-pull structures.
  \begin{enumerate}
    \item If $\li c$ and $\li c'$ are quasi-right-representable, then $\li (c
    \comp c')$ is quasi-right-representable.
    
    \item If $Ud$ and $Ud'$ are quasi-left-representable, then $U(d \comp d')$
    is quasi-left-representable.

  \end{enumerate}
  
\end{lemma}
\begin{proof}
  \begin{enumerate}
    \item We show that $\li (c \comp c')$ is quasi-right-representable.
    
    First, by Lemma \ref{lem:Fcc-equiv-FcFc'} there is a square $\alpha$ of the form

    \[\begin{tikzcd}[ampersand replacement=\&]
      {\li A} \&\& {\li A''} \\
      {\li A} \& {\li A'} \& {\li A''}
      \arrow["{\li (c \comp c')}", "\shortmid"{marking}, no head, from=1-1, to=1-3]
      \arrow[""{name=0, anchor=center, inner sep=0}, "{i_{\li A}(\delta^l)}"', from=1-1, to=2-1]
      \arrow[""{name=1, anchor=center, inner sep=0}, "{i_{\li A''}(\delta^r)}", from=1-3, to=2-3]
      \arrow["{\li c}", "\shortmid"{marking}, no head, from=2-1, to=2-2]
      \arrow["{\li c'}", "\shortmid"{marking}, no head, from=2-2, to=2-3]
      \arrow["\alpha"{description}, draw=none, from=0, to=1]
    \end{tikzcd}\]
    where $\delta^l$ and $\delta^r$ are syntactic perturbations.

    We define the projection $p_{\li (c \comp c')}$ to be $p_{\li c \comp \li c'} \circ i_{A''}(\delta^r)$.
    We define $\dellp_{\li (c \comp c')}$ to be $\dellp_{\li c \comp \li c'} \cdot \delta^l$.

    Then we can build the $\dnr$ square by pasting the square $\alpha$ on top of
    the $\dnr$ square for the composition $\li c \comp \li c'$, as shown below:

    \[\begin{tikzcd}[ampersand replacement=\&]
      {\li A} \&\& {\li A''} \\
      {\li A} \& {\li A'} \& {\li A''} \\
      {\li A} \&\& {\li A}
      \arrow["{\li (c \comp c')}", "\shortmid"{marking}, no head, from=1-1, to=1-3]
      \arrow[""{name=0, anchor=center, inner sep=0}, "{i_{\li A}(\delta^l)}"', from=1-1, to=2-1]
      \arrow[""{name=1, anchor=center, inner sep=0}, "{i_{\li A''}(\delta^r)}", from=1-3, to=2-3]
      \arrow["{\li c}", "\shortmid"{marking}, no head, from=2-1, to=2-2]
      \arrow[""{name=2, anchor=center, inner sep=0}, "{i_A(\dellp_{\li c \comp \li c'})}"', from=2-1, to=3-1]
      \arrow["{\li c'}", "\shortmid"{marking}, no head, from=2-2, to=2-3]
      \arrow[""{name=3, anchor=center, inner sep=0}, "{p_{\li c \comp \li c'}}", from=2-3, to=3-3]
      \arrow["{r(\li A)}"', "\shortmid"{marking}, no head, from=3-1, to=3-3]
      \arrow["\alpha"{marking, allow upside down}, draw=none, from=0, to=1]
      \arrow["{\dnr_{\li c \comp \li c'}}"{description}, draw=none, from=2, to=3]
    \end{tikzcd}\]

    We define $\delrp_{\li (c \comp c')}$ to be $\delrp_{\li c \comp \li c'} \cdot \delta^r$.
    For $\dnl$, we paste the identity square $\delta^r \ltdyn \delta^r$ on top of
    the $\dnl$ square for the composition $\li c \comp \li c'$, and below that we paste
    the square $\id \ltdyn_{\li (c \comp c')}^{\li c \comp \li c'} \id$ which we get from
    the fact that $\li$ is lax.

    \[\begin{tikzcd}[ampersand replacement=\&]
      {\li A''} \&\& {\li A''} \\
      {\li A''} \&\& {\li A''} \\
      {\li A} \& {\li A'} \& {\li A''} \\
      {\li A} \&\& {\li A''}
      \arrow["{r(\li A'')}", "\shortmid"{marking}, no head, from=1-1, to=1-3]
      \arrow[""{name=0, anchor=center, inner sep=0}, "{i_{A''}(\delta^r)}"', from=1-1, to=2-1]
      \arrow[""{name=1, anchor=center, inner sep=0}, "{i_{A''}(\delta^r)}", from=1-3, to=2-3]
      \arrow["{r(\li A'')}", "\shortmid"{marking}, no head, from=2-1, to=2-3]
      \arrow[""{name=2, anchor=center, inner sep=0}, "{p_{\li c \comp \li c'}}"', from=2-1, to=3-1]
      \arrow[""{name=3, anchor=center, inner sep=0}, "{i_{A''}(\delrp_{\li c \comp \li c'})}", from=2-3, to=3-3]
      \arrow["{\li c}", "\shortmid"{marking}, no head, from=3-1, to=3-2]
      \arrow["\id"', from=3-1, to=4-1]
      \arrow["{\li c'}", "\shortmid"{marking}, no head, from=3-2, to=3-3]
      \arrow["\id", from=3-3, to=4-3]
      \arrow["{\li (c \comp c')}"', "\shortmid"{marking}, no head, from=4-1, to=4-3]
      \arrow["\id"{marking, allow upside down}, draw=none, from=0, to=1]
      \arrow["{\dnl_{\li c \comp \li c'}}"{description}, draw=none, from=2, to=3]
    \end{tikzcd}\]

    \item The proof that $U(d \comp d')$ is quasi-left-representable is analogous.

  \end{enumerate}
\end{proof}


\begin{lemma}\label{lem:representation-product} 
  Let $c_1 : A_1 \rel A_1'$ and $c_2 : A_2 \rel A_2'$. 
  \begin{enumerate}
    \item If $c_1$ and $c_2$ are
    quasi-left-representable, then $c_1 \times c_2$ is quasi-left-representable.
    \item If $\li c_1$ and $\li c_2$ are quasi-right-representable, then so
    is $\li (c_1 \times c_2)$.
  \end{enumerate}


\end{lemma}
\begin{proof}
  \begin{enumerate}
    \item To show $c_1 \times c_2$ is quasi-left-representable, we define
    \begin{itemize}
      \item $e_{c_1 \times c_2} = e_{c_1} \times e_{c_2}$
      \item $\delre_{c_1 \times c_2} = \inl (\delre_{c_1}) \cdot \inr
      (\delre_{c_2})$ (recall from Section \ref{sec:perturbation-constructions}
      that $M_{A_1' \times A_2'}$ is the coproduct $M_{A_1'} \oplus M_{A_2'}$)
      \item $\delle_{c_1 \times c_2}$ is defined similarly
      \item The $\upl$ square is constructed using the functorial action of
      $\times$ and the corresponding squares for $c_1$ and $c_2$:
        \[\begin{tikzcd}[ampersand replacement=\&]
          {A_1 \times A_2} \& {A_1' \times A_2'} \\
          {A_1' \times A_2} \& {A_1' \times A_2'} \\
          {A_1' \times A_2'} \& {A_1' \times A_2'}
          \arrow["{c_1 \times c_2}", "\shortmid"{marking}, no head, from=1-1, to=1-2]
          \arrow[""{name=0, anchor=center, inner sep=0}, "{e_{c_1} \times \id}"', from=1-1, to=2-1]
          \arrow[""{name=1, anchor=center, inner sep=0}, "{i_{A_1'}(\delre_{c_1}) \times \id}", from=1-2, to=2-2]
          \arrow["{r(A_1') \times c_2}", "\shortmid"{marking}, no head, from=2-1, to=2-2]
          \arrow[""{name=2, anchor=center, inner sep=0}, "{\id \times e_{c_2}}"', from=2-1, to=3-1]
          \arrow[""{name=3, anchor=center, inner sep=0}, "{\id \times i_{A_2'}(\delre_{c_2})}", from=2-2, to=3-2]
          \arrow["{r(A_1') \times r(A_2')}"', "\shortmid"{marking}, no head, from=3-1, to=3-2]
          \arrow["{\upl_{c_1} \times \id}"{description}, draw=none, from=0, to=1]
          \arrow["{\id \times \upl_{c_2}}"{description}, draw=none, from=2, to=3]
        \end{tikzcd}\]
      \item The $\upr$ square is defined similarly
    \end{itemize}

    \item To show that $\li(c_1 \times c_2)$ is quasi-right-representable, we define
    \begin{itemize}
      \item $p_{\li (c_1 \times c_2)} = (p_{\li c_1} \timesk A_2) \circ (A_1' \timesk p_{\li c_2})$
      \item $\dellp_{\li (c_1 \times c_2)} = (\dellp_{\li c_1} \timesk \id) \cdot (\id \timesk \dellp_{\li c_2})$ 
      using the Kleisli product actions on syntactic perturbations in Definition \ref{def:kleisli-product-perturbations}
      \item $\delrp_{\li (c_1 \times c_2)} = (\delrp_{\li c_1} \timesk \id) \cdot (\id \timesk \delrp_{\li c_2})$
      \item The squares are obtained via the functorial action of $\timesk$ on the squares for $\li c_1$ and $\li c_2$.
    \end{itemize}
  \end{enumerate}
\end{proof}

\begin{lemma}\label{lem:representation-arrow}
  Let $c : A \rel A'$ and $d : B \rel B'$.
  \begin{enumerate}
    \item If $c$ is quasi-left-representable and $d$ is quasi-right-representable, then
    $c \arr d$ is quasi-right-representable.
  
    \item If $\li c$ is quasi-right-representable and $Ud$ is
    quasi-left-representable, then $U(c \arr d)$ is quasi-left-representable.
  \end{enumerate}
  

\end{lemma}
\begin{proof}
  \begin{enumerate}
    \item 
      To show $c \arr d$ is quasi-right-representable, we define:
      \begin{itemize}
        \item $p_{c \arr d} = e_c \arr p_d$
        \item $\dellp_{c \arr d} = \inl(\delle_c) \cdot \inr(\dellp_d)$
        \item $\delrp_{c \arr d} = \inl(\delre_c) \cdot \inr(\delrp_d)$
        \item The squares $\dnr$ and $\dnl$ are obtained via the functorial action of $\arr$ on squares.
        
      \end{itemize}

    \item
      We show $U(c \arr d)$ is quasi-left-representable as follows:
      \begin{itemize}
        \item $e_{U(c \arr d)} = (p_{\li c} \tok B') \circ (A \tok e_{Ud})$
        \item $\delre_{U(c \arr d)} = (\delrp_{\li c} \tok B') \cdot (A' \tok
              \delre_{Ud})$ using the Kleisli arrow actions on syntactic
              perturbations in Definition \ref{def:kleisli-arrow-perturbations}.
        \item $\delle_{U(c \arr d)} = (\dellp_{\li c} \tok B) \cdot (A \tok \delle_{Ud})$
        \item The squares $\upl$ and $\upr$ are obtained via the functorial action of $\tok$ on squares.
        For instance, $\upl$ is given by the following square:

        \[\begin{tikzcd}[ampersand replacement=\&,row sep=large]
          {U(A \to B)} \&\& {U(A' \to B')} \\
          {U(A \to B')} \&\& {U(A' \to B')} \\
          {U(A' \to B')} \&\& {U(A' \to B')}
          \arrow["{U(c \to d)}", from=1-1, to=1-3]
          \arrow["{U(c \to r(B'))}", from=2-1, to=2-3]
          \arrow["{U(r(A') \to r(B'))}", from=3-1, to=3-3]
          \arrow[""{name=0, anchor=center, inner sep=0}, "{p_{Fc} \tok B'}"', from=2-1, to=3-1]
          \arrow[""{name=1, anchor=center, inner sep=0}, "{A' \tok \delre_{Ud}}", from=1-3, to=2-3]
          \arrow[""{name=2, anchor=center, inner sep=0}, "{\delrp_{Fc} \tok B'}", from=2-3, to=3-3]
          \arrow[""{name=3, anchor=center, inner sep=0}, "{A \tok e_{Ud}}"', from=1-1, to=2-1]
          \arrow["{\id_{Fc} \tok \upl_{Ud}}"{description}, draw=none, from=3, to=1]
          \arrow["{\dnl_{Fc} \tok \id_{r(B')}}"{description}, draw=none, from=0, to=2]
        \end{tikzcd}\]
      \end{itemize}

    The construction of $\upr$ is similar.
\end{enumerate}
\end{proof}

\subsection{Model Construction: Composition and Functorial Actions on Relations}

We now show that our notions of value and computation relations compose.
\begin{definition}[composition of value (resp. computation) relations]\label{def:value-computation-rel-comp}
  Let $A$, $A'$ and $A''$ be value objects and let $c$ and $c'$ be \emph{value
  relations} between $A$ and $A'$ and $A'$ and $A''$ respectively. Then we
  define their composition $cc'$ as follows:

  \begin{itemize}
    \item The predomain relation is the composition of the underlying predomain
    relations $cc'$
    \item The push-pull structure is given by Lemma \ref{lem:push-pull-comp}.
    \item Quasi-left-representability of $cc'$ follows from Lemma
    \ref{lem:representation-comp} and the fact that $c$ and $c'$ are
    quasi-left-representable
    \item Quasi-right-representability of $\li(cc')$ holds by Lemma
    \ref{lem:representation-comp-F-U} and the quasi-right-representability of
    $\li c$ and $\li c'$
  \end{itemize}

  Likewise, let $B$, $B'$ and $B''$ be computation objects and let $d$ and $d'$ be
  \emph{computation relations} between $B$ and $B'$ and $B'$ and $B''$
  respectively. We define their composition $dd'$ as follows:
  
  \begin{itemize}
    \item The error domain relation is the composition of the underlying error
    domain relations $dd'$
    \item The push-pull structure is given by Lemma \ref{lem:push-pull-comp}.
    \item Quasi-right-representability of $dd'$ follows from Lemma
    \ref{lem:representation-comp} and the fact that $d$ and $d'$ are
    quasi-right-representable
    \item Quasi-left-representability of $U(dd')$ holds by Lemma
    \ref{lem:representation-comp-F-U} and the quasi-left-representability of
    $Ud$ and $Ud'$
  \end{itemize}
\end{definition}

We can define the functorial actions of $\li$, $U$, $\times$, and $\arr$ on
value and computation relations as follows:
\begin{definition}[functorial actions on relations]\label{def:functorial-actions-on-relations}
  
  \begin{enumerate}

    \item Let $A$ and $A'$ be value objects and $c : A \rel A'$ a value relation.
          Then we define the computation relation $\li c$ as follows:
      \begin{itemize}
        \item The error domain relation is given by the action of $\li$ on the
        underlying predomain relation of $c$
        \item The push-pull structure is given by Lemma \ref{lem:push-pull-U-F}
        \item Quasi-right-representability of $\li c$ holds because it is part
        of the definition of the value relation $c$
        \item Quasi-left-representability of $U \li c$ follows from Lemma
        \ref{lem:representation-U-F} and the quasi-left-representability of $c$
      \end{itemize}
          
    \item Let $B$ and $B'$ be computation objects and $d : B \rel B'$ a computation relation.
          We define $Ud$ as follows:
          \begin{itemize}
            \item The predomain relation is given by the action of $U$ on the
            error domain relation of $d$
            \item The push-pull structure is given by Lemma
            \ref{lem:push-pull-U-F}
            \item Quasi-left-representability of $Ud$ holds because it is part
            of the definition of the computation relation $d$
            \item Quasi-right-representability of $\li(Ud)$ follows from Lemma
            \ref{lem:representation-U-F} and the quasi-right-representability of
            $d$
          \end{itemize}

    \item Let $c_1 : A_1 \rel A_1'$ and $c_2 : A_2 \rel A_2'$ be value relations.
          We define $c_1 \times c_2$ as follows:
          \begin{itemize}
            \item The relation is given by the action of $\times$ on the
            predomain relations of $c_1$ and $c_2$
            \item The push-pull structure is given by Lemma
            \ref{lem:push-pull-times}
            \item Quasi-left-representability of $c_1 \times c_2$ and
            quasi-right-representabiity of $\li(c_1 \times c_2)$ are given by
            Lemma \ref{lem:representation-product}
          \end{itemize}

    \item Let $c : A \rel A'$ be a value relation and $d : B \rel B'$ be a computation relation.
          We define $c \arr d$ as follows:
          \begin{itemize}
            \item The error domain relation is given by the action of $\arr$ on
            the underlying relations of $c$ and $d$
            \item The push-pull structure is given by Lemma
            \ref{lem:push-pull-arrow}
            \item Quasi-right-representability of $c \arr d$, and
                  quasi-left-representability of $U(c \arr d)$ follow from Lemma
                  \ref{lem:representation-arrow}
          \end{itemize}

  \end{enumerate}
\end{definition}

Lastly, we prove quasi-equivalence of relations involving composition and the
functors $U$, $\li$, $\to$ and $\times$.

\begin{lemma}\label{lem:quasi-order-equiv-functors}
  The following hold:
  \begin{itemize}
    \item $U(d \comp d') \qordeq U(d) \comp U(d')$
    \item $\li(c \comp c') \qordeq \li (c) \comp \li (c')$
    \item $(c \comp c') \to (d \comp d') \qordeq (c \to d) \comp (c' \to d')$
    \item $(c_1 \comp c_1') \times (c_2 \comp c_2') \qordeq (c_1 \times c_2) \comp (c_1'\times c_2')$
  \end{itemize}
\end{lemma}
\begin{proof}
  (1) and (2) were already shown in Lemmas \ref{lem:Fcc-equiv-FcFc'} and
  \ref{lem:Udd-equiv-UdUd'}. (3) follows from observing that both relations are
  right-represented by the morphism $e_c'e_c \arr p_{d}p_{d'}$ and (4) from the
  fact that both are left-represented by the morphism $e_{c_1'} e_{c_1} \times
  e_{c_2'} e_{c_2}$
\end{proof}


\subsection{Definitions of Error Ordering and Weak Bisimilarity for the Delay Monad}
\label{sec:relations-on-delay-monad}

We define a notion of lock-step error ordering and weak bisimilarity relation
for the coinductive delay monad $\delay(\mathbb{N} + {\mho})$:

The lock-step error ordering is defined coinductively by the following rules:

\begin{mathpar}
  \inferrule*[]
  { }
  {\tnow (\inr\, 1) \ledelay d}

  \inferrule*[]
  {x_1 \le_X x_2}
  {\tnow (\inl\, x_1) \ledelay \tnow (\inl\, x_2)}

  \inferrule*[]
  {d_1 \ledelay d_2}
  {\tlater\, d_1 \ledelay \tlater\, d_2}
\end{mathpar}
And we similarly define by coinduction a ``weak bisimilarity'' relation on
$\delay(\mathbb{N} + {\mho})$. This uses a relation $d \Da x_?$ between
$\delay(\mathbb{N} + {\mho})$ and $\mathbb{N} + {\mho}$ that is defined as $d
\Da n_? := \Sigma_{i \in \mathbb{N}} d = \tlater^i(\tnow\, n_?)$. Then weak
bisimilarity for the delay monad is defined coinductively by the rules
\begin{mathpar}
  \inferrule*[]
  {n_? \bisim_{\mathbb{N} + {\mho}} m_?}
  {\tnow\, n_? \bisimdelay \tnow\, m_? }

  \inferrule*[leftskip=1.5em]
  {d_1 \Da n_? \and n_? \bisim_{\mathbb{N} + {\mho}} m_?}
  {\tlater\, d_1 \bisimdelay \tnow\, m_? }

  \inferrule*[leftskip=1.5em]
  {d_2 \Da m_? \and n_? \bisim_{\mathbb{N} + {\mho}} m_?}
  {\tnow\, n_? \bisimdelay \tlater\, d_2}

  \inferrule*[leftskip=1.5em]
  {d_1 \bisimdelay d_2}
  {\tlater\, d_1 \bisimdelay \tlater\, d_2 }
  %
  %
\end{mathpar}

\end{document}